\documentclass[12pt,a4paper]{article}
\usepackage{setspace}
\usepackage[utf8]{inputenc}
\usepackage[T1]{fontenc}
\usepackage[english]{babel}
\usepackage{amsfonts}
\usepackage{amssymb}
\usepackage{amsmath}
\usepackage{graphicx}
\usepackage{stmaryrd}
\usepackage{geometry}
\usepackage{setspace}
\usepackage{amsthm}
\usepackage{alltt}
\usepackage{booktabs}
\usepackage{todonotes}
\usepackage{caption}
\usepackage{float}
\usepackage{algorithm}
\usepackage{algorithmicx}
\usepackage{algpseudocode}
\usepackage{epsfig}
\usepackage{pifont}
\usepackage{array}

\newtheorem{rem}{Remark}

\newtheorem{prop}{Proposition}

\def\RR{\textsf{R}\/}

\let\pkg=\strong

\usepackage{natbib}
\setcitestyle{aysep={}} 
\usepackage[pdftex]{hyperref}
\hypersetup{colorlinks=true,linkcolor=blue,citecolor=magenta}

\graphicspath{{Images/}}

\title{Applicability and Interpretability of Hierarchical Agglomerative Clustering With or Without Contiguity Constraints}

\author{Nathanaël Randriamihamison$^{1,2,3}$, Nathalie Vialaneix$^1$ \& Pierre 
Neuvial$^2$}

\date{$^1$ MIAT, Université de Toulouse, INRA, Castanet Tolosan, France\\
\{nathanael.randriamihamison, nathalie.vialaneix\}@inra.fr\\
$^2$ Institut de Mathématiques de Toulouse; 
UMR 5219, Université de Toulouse, CNRS
UPS, F-31062 Toulouse Cedex 9, France\\
pierre.neuvial@math.univ-toulouse.fr\\
$^3$ INRIA Bordeaux Sud-Ouest, CQFD Team, France}

\begin{document}

% \tableofcontents
% \newpage

\maketitle

\begin{abstract}
  Hierarchical Agglomerative Classification (HAC) with Ward's linkage has been
  widely used since its introduction by \cite{ward_JASA1963}. This article
  reviews extensions of HAC to various input data and contiguity-constrained
  HAC, and provides applicability conditions. In addition, various versions of
  the graphical representation of the results as a dendrogram are also presented
  and their properties are clarified. We clarify and complete the results
  already available in an heterogeneous literature using a uniform
  background. In particular, this study reveals an important distinction between
  a consistency property of the dendrogram and the absence of crossover within
  it. Finally, a simulation study shows that the constrained version of HAC can
  sometimes provide more relevant results than its unconstrained version despite
  the fact that the latter optimizes the objective criterion on a reduced set of
  solutions at each step.  Overall, the article provides comprehensive
  recommendations for the use of HAC and constrained HAC depending on the input
  data as well as for the representation of the results.
\end{abstract}


\section{Introduction}

Hierarchical Agglomerative Classification (HAC) with Ward's linkage has been
widely used since its introduction by \cite{ward_JASA1963}. The method is
appealing since it provides a simple approach to approximate, for any given
number of clusters, the partition minimizing the within-cluster inertia or
``error sum of squares''. In addition to its simplicity and the fact that it is
based on a natural quality criterion, HAC often comes with a popular graphical
representation called a dendrogram, that is used as a support for model
selection (choice of the number of clusters) and result interpretation.
Originally described to cluster data in $\mathbb{R}^p$, the method has been
applied more generally to data described by arbitrary distances (or
dis-similarities). Constrained versions of HAC have also been proposed to
incorporate a ``contiguity'' relation between objects into the clustering
process \citep{lebart_CAD1978,Grimm1987}.

However, as already shown by \cite{Murtagh2014}, confusions still exist between
the different versions and how the results are represented with a dendrogram, 
which is also illustrated in \citep{Grimm1987} that presents different 
alternatives for the representation. These have resulted in different 
implementations of the Ward's clustering algorithm, with notable differences in 
the results. More importantly, the applicability framework of the different 
versions is not always clear: \cite{Batagelj1981} has given very general 
necessary and sufficient conditions on a linkage value to ensure that it is 
always increasing for any given dissimilarity. This property is important to 
ensure the consistency between the results of HAC and their graphical display as
a dendrogram. Conditions on a general constraint are also provided in
\cite{Ferligoj1982} to ensure a similar property and \cite{Grimm1987} proposes
alternative solutions to the standard heights to address the fact that the
linkage might sometimes fail to provide a consistent representation of the
results of HAC. However, none of these articles fully cover the theoretical
properties of these alternatives, for unconstrained and constrained versions of
the method.

The scope of the present article is to clarify the conditions of applicability
and interpretability of the different versions of HAC and contiguity-constrained
HAC (CCHAC). We discuss the relevance of using them with different types of
input data and how the corresponding results can be interpreted. We perform a
systematic study of the monotonicity of the different versions of the dendrogram
heights by reporting the results already available in the literature for
standard HAC and its various extensions and by completing the ones that were not
available to our knowledge. In addition to providing a uniform presentation of
a number of results partially present in the literature, this study reveals an
important distinction between this consistency property and the absence of
crossover within the dendrogram that was not discussed before.

Finally, we illustrate the respective behavior of HAC and CCHAC in a simulation 
study where the different heights are used in order to represent the results. 
This simulation shows that, in addition to reducing the computational time 
needed to perform the method, the constrained version (CCHAC) can also provide 
better solution than the standard one (HAC) when the constraint is consistent 
with the data, despite the fact that it optimizes the objective criterion on a 
reduced set of solutions at each step.

\section{HAC and contiguity-constrained HAC}
\label{sec:hac-cchac}

\subsection{Hierarchical Agglomerative Clustering} 
\label{sec:hac-euclidean-framework}

HAC was initially described by \citet{ward_JASA1963} for data in
$\mathbb{R}^p$. Let $\Omega := \{x_1,\cdots,x_n\}$ be the set of objects to be
clustered, which are assumed to lie in $\mathbb{R}^p$. A cluster is a subset of
$\Omega$. The loss of information when grouping objects into a cluster
$G \subset \Omega$ is quantified by the inertia (also known as \emph{Error Sum
  of Squares}, ESS):
\begin{equation}
  \label{eq:inertia}
  I(G) = \sum\limits_{x_i \in G} \|x_i - \bar{x}_G\|_{\mathbb{R}^p}^2\,,
\end{equation}
where $\bar{x}_G = n^{-1}\sum_{i=1}^n x_i$ is the center of gravity of $G$. 
Starting from a partition $\mathcal{P}=\{ G_1, \cdots, G_{l}\}$ of $\Omega$, the loss of information when merging two clusters $G_u$ and $G_v$ of $\mathcal{P}$ is quantified by :
\begin{equation}
\label{eq:ward-s-linkage}
\delta(G_u,G_v):=I(G_u \cup G_v)-I(G_u)-I(G_v).
\end{equation}
The quantity $\delta$ is known as Ward's linkage and it is equal to the 
variation of within-cluster inertia (also called \emph{within-cluster sum of 
squares}) after merging two clusters. It also corresponds to the squared distance between centers of gravity:
\begin{equation}
\label{eq:ward-s-linkage-gravity}
\delta(G_u,G_v)= \frac{|G_u||G_v|}{|G_u|+|G_v|}\|\bar{x}_{G_u} - \bar{x}_{G_v}\|^2_{\mathbb{R}^p}.
\end{equation}

The HAC algorithm is described in Algorithm~\ref{algo:HAC}.
Starting from the trivial partition 
$\mathcal{P}_1=\{\{x_1\},\{x_2\},\cdots,\{x_n\}\}$ with $n$ singletons, the HAC 
algorithm creates a sequence of partitions by successively merging the two 
clusters whose linkage $\delta$ is the smallest, until all objects have been 
merged into a single cluster.
\begin{algorithm}
\caption{Standard Hierarchical Agglomerative Clustering (HAC)}
	\begin{algorithmic}[1]
	\State \textbf{Initialization:} 
$\mathcal{P}_1=\{\{x_1\},\{x_2\},\cdots,\{x_n\}\}$
	\For{$t=1$ to $n-1$}
		\State Compute all pairwise linkage values between clusters of the current 
partition $\mathcal{P}_{t}$
		\State Merge the two clusters with minimal linkage value to obtain the 
next partition $\mathcal{P}_{t+1}$
	\EndFor
	\State \Return $\{\mathcal{P}_{1}, \mathcal{P}_{2}, \ldots, \mathcal{P}_{n}\}$
	\end{algorithmic}
\label{algo:HAC}
\end{algorithm}
Linkage values at step $t$ can be efficiently updated using linkage values at 
step $t-1$ with a formula known as the Lance-Williams formula \citep{Lance1967}. 
In the case of Ward's linkage, this formula has first been demonstrated by 
\cite{Wishart1969}:
\begin{eqnarray}
\label{eq:lance-williams_ward}
  \delta(G_{u} \cup G_{v}, G_w) &=& 
\frac{|G_{u} |+|G_w|}{|G_{u}|+|G_{v} |+|G_w|} \delta(G_{u} ,G_w) + 
\frac{|G_{v} |+|G_w|}{|G_{u}|+|G_{v} |+|G_w|}  \delta(G_{v} ,G_w) \notag\\
    && \qquad -\frac{|G_w|}{|G_{u} |+|G_{v} |+|G_w|} \delta(G_{u} , G_{v})
\end{eqnarray}
where $|G|$ denotes the cardinal of any cluster $G$. 

The framework of the current section can be extended straightforwardly to the 
case where the objects to cluster are weighted. However, this study focuses on 
uniform weights for the sake of simplicity.

\subsection{HAC under contiguity constraint}
\label{sec:CCHAC}

In various applications, there exists an \emph{a priori} information about the
relations between the objects. For instance, it is the case for spatial
statistics, where objects possess natural proximity relations, in genomics,
where genomic loci are linearly ordered along the chromosome, or in
neuroimaging, with the three-dimensional structure of the brain. According to
this point of view, Contiguity-Constrained HAC (CCHAC) allows only mergers
between contiguous objects. Considering this approach can have two benefits: (i)
more interpretable results by taking into account the natural structure of the
data; (ii) a decreased computational time, because only a subset of all possible
mergers are considered.

A very general framework for constrained HAC is described in 
\cite{Ferligoj1982}: the contiguity is defined by an arbitrary symmetric 
relation $\mathcal{R} \subset \Omega \times \Omega$ that indicates which pairs 
of objects are said \emph{contiguous}. Only these pairs are then allowed to be 
merged at the first step of the algorithm, using the same objective function 
than in the standard HAC algorithm. The next step iterates 
similarly, by using the following rule to extend the contiguity relation to 
merged clusters:
\[
  \left(G_u \cup G_v , G_w \right) \in \mathcal{R} \qquad \Leftrightarrow 
\qquad \left(G_{u} , G_w\right) \in \mathcal{R} \textrm{ or } \left(G_v ,
G_w\right) \in \mathcal{R}.
\]
Algorithm~\ref{algo:CCHAC} describes contiguity-constrained hierarchical
agglomerative clustering (CCHAC).
\begin{algorithm}
\caption{Contiguity-Constrained Hierarchical Agglomerative Clustering (CCHAC)}
	\begin{algorithmic}[1]
	\State \textbf{Initialization:} 
$\mathcal{P}_1=\{G_1^1,G_2^1,\cdots,G_n^1\}$ where $G_u^1 = \{x_u\}$. 
Contiguous singletons are defined by $\mathcal{R}_1 = \mathcal{R} \subset 
\Omega \times \Omega$.
	\For{$t=1$ to $n-1$}
		\State Compute all pairwise linkage values between \emph{contiguous} 
clusters of the current partition $\mathcal{P}_t$ with respect to 
$\mathcal{R}_t$
		\State Merge the two \emph{contiguous} clusters, $G^t_{v_1}$ and 
$G^t_{v_2}$ with minimal linkage value to obtain the next partition 
$\mathcal{P}_{t+1} = \{G_u^{t+1}\}_{u=1,\ldots,n-t-1}$
    \State Extend the contiguity relation to the new cluster $G^t_{v_1} \cup 
G^{t}_{v_2} \in \mathcal{P}_{t+1}$ by setting 
  \[
    \left(G^t_{v_1} \cup G^t_{v_2} , G^{t}_w \right) \in \mathcal{R}_{t+1} 
\qquad \Leftrightarrow \qquad \left(G^t_{v_1} , G^{t}_w\right) \in 
\mathcal{R}_t 
\textrm{ or } \left(G^t_{v_2} , G^{t}_w\right) \in \mathcal{R}_t.
  \]
	\EndFor
	\State \Return $\{\mathcal{P}_{1}, \mathcal{P}_{2}, \ldots, \mathcal{P}_{n}\}$
	\end{algorithmic}
\label{algo:CCHAC}
\end{algorithm}
The only difference with standard HAC lies in the fact that only contiguous clusters are
merged. From a computational viewpoint, only the linkage values for a
subset of $\mathcal{P}_t \times \mathcal{P}_t$ have to be considered,
which can drastically reduce the number of values to be computed with
respect to the standard algorithm. This gain in computational time
comes at the price of a (potential) loss in the objective function
at a given step of the algorithm, especially if the constraint is not
consistent with the dissimilarity or similarity values (see
Section~\ref{sec:simulations} for illustration and discussion). This
also has a side effect on standard representations of the result of
the algorithm, which is discussed in Section~\ref{sec:monoticity}. 

\paragraph{Order-constrained HAC.} A simple and useful case of contiguity
constraint is the case when the symmetric relation is a contiguity relation
defined along a line.  This special and simple case is frequently encountered in
genomics (where the contiguity thus corresponds to the genomic position on a
given chromosome) and will be called \emph{order-constrained HAC} (OCHAC) in the
sequel. In this specific case, every cluster has exactly two neighbours (except
for the two positioned at the beginning and the end of the line) and at step $t$
of the algorithm, only $n-t$ values of the linkage have to be computed (instead
of $(n-t)(n-t-1)/2$ for standard HAC). This case is the one implemented in the
\RR{} package \pkg{adjclust} and an efficient algorithm is described in
\cite{ambroise_etal_p2018} for sparse datasets. In this paper, we demonstrate
the good properties of the CCHAC for the case of a general contiguity relation
and illustrate the opposite situation (where some good properties are not always
satisfied for CCHAC) by providing counter-examples and illustrations in the
specific case of OCHAC.

\section{Validity of HAC in possibly non-Euclidean settings}
\label{sec:HAC_extensions}

In this section, we systematically justify the use of HAC algorithm (with or without contiguity constraints) for all 
kinds of proximity data, including dissimilarity and similarity data. 

\subsection{Extension to dissimilarity data}
\label{sec:extens-diss-data}

The HAC algorithm of \citet{ward_JASA1963} has been designed to cluster 
elements of $\mathbb{R}^p$. In practice however, the objects to be 
clustered are often only indirectly described by a matrix of pairwise 
dissimilarities, $D=(d_{ij})_{1\leq i,j \leq n}$. Formally, a dissimilarity is 
a generalization of a distance, that can not necessarily be embedded into an 
Euclidean space (\textit{e.g.}, because the triangle inequality does not 
hold for instance). Here, we only assume that $D$ satisfies the following properties 
for all $i,j \in \{1, \dots, n\}$:
\begin{equation*}
  d_{ij} \geq 0; \quad d_{ii}=0; \quad d_{ij}=d_{ji}\,.
\end{equation*}

The HAC algorithm will be applicable to such a dissimilarity matrix $D$ if $D$
is Euclidean. Formally, $D$ is Euclidean if there exists an Euclidean space
$(E, \langle \cdot , \cdot \rangle )$ and $n$ points
$\{ x_1, ... , x_n \} \subset E$ such that $d_{ij} = \| x_i-x_j \|$ for all
$i,j \in \{1, ... , n \}$, with $\| \cdot \|$ the norm induced by the inner
product, $\langle . , . \rangle$, on $E$. Under this assumption, the
dissimilarity case is a simple extension of the original $\mathbb{R}^p$
framework described in Section~\ref{sec:hac-cchac}.  Different versions of a
necessary and sufficient condition for which an observed dissimilarity matrix is
Euclidean have been obtained in \cite{schoenberg_AM1935,Young1938,Krislock2012}.

When such conditions do not
hold, $D$ is simply called a dissimilarity dataset, which is a particular case
of proximity or relational datasets. \cite{schleif_tino_NC2015} have proposed a
typology of such datasets and described different approaches that can be
used to extend statistical or learning methods defined for Euclidean data to
such proximity data. In brief, the first main strategy consists in finding a way
to turn a non Euclidean dissimilarity into an Euclidean distance, that is the
closest (in some sense) to the original dissimilarity. This can be performed
using eigenvalue corrections \citep{chen_etal_JMLR2009}, embedding strategies
(like multidimensional scaling, \cite{Kruskal1964}) or solving a maximum
alignment problem \citep{chen_ye_ICML2008}, for instance.

\paragraph{A general construction.} Alternatively, by using an analogy between distance and dissimilarity, HAC can
be directly extended to non-Euclidean data as in
\cite{chavent_etal_CS2018}. This extension stems from the fact that, in the
Euclidean case of Section~\ref{sec:hac-cchac}, the inertia of a clusters may be expressed only in function of sums of the entries of the pairwise distances $(\| x_i - x_j \|_{\mathbb{R}^p}, 1\leq i,j \leq n)$:
  \begin{align}
    \label{eq:inertia-dist}
    I(G) = \frac{\Delta(G,G)}{2|G|} \,,
  \end{align}
where $\Delta$ is defined by $\Delta(G_u, G_v) = \sum_{x_i \in G_u,x_{j} \in 
G_v} \| x_i - x_j \|_{\mathbb{R}^p}^2$ for any clusters $G_u$ and $G_v$.
As a consequence of~\eqref{eq:inertia-dist}, Ward's linkage between any two clusters $G_u$ and $G_v$ may be itself be written in function of these pairwise distances, see, e.g., \citet[p. 279]{Murtagh2014}:
\begin{align}
  \label{eq:Ward-dist}
  \delta(G_u, G_v) &=  \frac{|G_u||G_v|}{|G_u|+|G_v|}\left(\frac{\Delta(G_u, G_v) }{|G_u||G_v|}- \frac{\Delta(G_u, G_u) }{2|G_u|^2} - \frac{\Delta(G_v, G_v) }{2|G_v|^2}\right)\,.
\end{align}
Therefore, as proposed by~\cite{chavent_etal_CS2018}, an elegant way to extend Ward's HAC to dissimilarity data is to \emph{define} the inertia of a cluster using~\eqref{eq:inertia-dist}, with (sums of) distances replaced by (sums of) dissimilarities, that is:
\begin{align}
  \label{eq:inertia-diss}
  \tilde{I}(G) = \frac{\tilde{\Delta}(G,G)}{2|G|} \,,
\end{align}
where 
\begin{align}
\label{eq:Delta-dissim}
\tilde{\Delta}(G_u, G_v) = \sum_{x_i \in G_u,x_j \in G_v}  d_{ij}^2\,.
\end{align}
The corresponding HAC is then formally obtained as the output of
Algorithm~\ref{algo:HAC}, as described in
Section~\ref{sec:hac-euclidean-framework}. In particular, Ward's linkage is
still given by \eqref{eq:Ward-dist}, with $\Delta$ formally replaced by
$\tilde{\Delta}$, and, as a consequence, the Lance-Williams update formula is also still given by
\eqref{eq:lance-williams_ward}. 
When the elements of $\Omega$ do belong to an Euclidean space and the
dissimilarities are the pairwise Euclidean distances
$\| x_i - x_j \|_{\mathbb{R}^p}$, these two definitions of HAC
coincide. Otherwise, HAC is still formally defined, and the linkage can still be
seen as a measure of heterogeneity, but the interpretation of a cluster inertia
as an average squared distance to the center of gravity of the cluster (as in
Equation~\eqref{eq:inertia}) is lost. Since the two definitions, $I$ and
$\tilde{I}$ coincide for the Euclidean case, we will only use the notation $I$
in the sequel for the sake of simplicity, even when the data are non Euclidean
dissimilarity data.

The above approach based on pairwise dissimilarities and pseudo-intertia may be used to recover generalizations of Ward-based HAC to non-Euclidean distances already proposed in the literature. In particular, the Ward HAC algorithm associated to $d_{ij} =  \| x_i - x_j \|^{\alpha/2}_{\mathbb{R}^p}$ for $0 < \alpha \leq 2$ and $d_{ij} =  \| x_i - x_j \|_{1, \mathbb{R}^p}$ (the latter is also called the Manhattan distance) correspond to the methods proposed by \cite{szekely_rizzo_JC2005} and \cite{strauss_vonmaltitz_PO2017}, respectively. 

\begin{rem}
  \cite{szekely_rizzo_JC2005} and \cite{strauss_vonmaltitz_PO2017} take a
  different point of view: they \emph{define} the linkage between two clusters
  by~\eqref{eq:Ward-dist} (up to a scaling factor 1/2); their generalized HAC is
  then the HAC associated to this linkage. Then, they \emph{prove} that the
  Lance-Williams Equation~\eqref{eq:lance-williams_ward} is still valid for this
  linkage. We favor the above construction by \cite{chavent_etal_CS2018}, which is simply based on \emph{pairwise} dissimilarities, as it is more intrinsic. It provides a justification for the linkage formula, and the Lance-Williams formula is automatically valid with no proof required. 
\end{rem}

Finally, there is an ambiguity in the definition of the pseudo-inertia as an
extension of the Ward's case. If most authors consider that the dissimilarity is
associated to a distance and therefore define the pseudo-inertia based on the
squared values $d_{ij}^2$, some authors (as \cite{strauss_vonmaltitz_PO2017})
define a linkage equal to the one that would have been obtained with Ward's 
linkage and a pseudo-inertia described as $\frac{1}{2|G|} \sum_{x_i, x_j \in G} 
d_{ij}$. This ambiguity has long been enforced by popular implemented versions 
of the algorithm, as it was the case in the \RR{} function \texttt{hclust} 
before \cite{Murtagh2014} raised and corrected this problem.

\subsection{Extension to kernel data} 
\label{sec:extens-kern-data}
In some cases, proximity relations between objects are described by their
resemblance instead of their dissimilarity. We start with the case when the data
are described by a kernel matrix. A kernel matrix is a symmetric
positive-definite matrix $K=(k_{ij})_{1\leq i,j \leq n}$ whose entry $k_{ij}$
corresponds to a measure of resemblance between $x_i$ and $x_j$. Here, contrary
to the Euclidean setting, no specific structure is assumed for $\Omega$, which
can be an arbitrary set.
 
\cite{Aronszajn1950} has proved that there exists a unique Hilbert space 
$\mathcal{H}$ equipped with the inner product $\langle \cdot , \cdot 
\rangle_\mathcal{H}$ and a unique map $\phi: \Omega \longrightarrow 
\mathcal{H}$, such that $k_{ij}= \langle\phi(x_i), \phi(x_j) 
\rangle_\mathcal{H}$. This allows to consider the associated distance in 
$\mathcal{H}$ between any two elements $\phi(x_i)$ and $\phi(x_j)$ for $x_i,x_j 
\in \Omega$, that implicitly defines a Euclidean distance in $\Omega$ by:
\[
	d_{ij}=d(x_i,x_j):=\|\phi(x_i)-\phi(x_j)\|_\mathcal{H} \,,
\]
so that 
\begin{equation}
\label{eq:kernel-squared-distance}
 d_{ij}^2= k_{ii}+k_{jj}-2k_{ij}\,.
\end{equation}
Therefore, it is possible to use Algorithm~\ref{algo:HAC} for kernel data, even
when $\mathcal{H}$ is not known explicitly and/or when it is not
finite-dimensional. This is an instance of the so-called ``kernel trick''
\citep{Scholkopf2002}. The associated Ward's linkage can itself be re-written
directly using sums of elements of the kernel matrix, as shown, e.g., in
\cite{Dehman2015}:
\begin{equation}
  \label{eq:ward_kernel}
  \delta(G_u,G_v)=\frac{|G_u||G_v|}{|G_u| + |G_v|} \left( 
\frac{R_{G_u,G_u}}{|G_u|^2} + \frac{R_{G_v,G_v}}{|G_v|^2} - 2 \frac{R_{G_u,G_v}}{|G_u||G_v|} \right)\,,
\end{equation}
where 
$R_{G_u,G_v}=\sum\limits_{(x_i,x_j) \in G_u\times G_v} k_{ij}$.
% Note that this is formally the same formula as \eqref{eq:Ward-dist} where $d_{ij}^2$ is replaced by $-2k_{ij}$ (although $-2k_{ij}$ cannot be interpreted as a squared dissimilarity because it does not satisfy the requirement that self-dissimilarities are null). Indeed, the diagonal terms $k_{ii}+k_{jj}$ cancel out when applying  \eqref{eq:Ward-dist} to the dissimilarity defined by \eqref{eq:kernel-squared-distance}.

Contrary to the dissimilarity case described in
Section~\ref{sec:extens-diss-data}, the kernel case is a truly interpretable
generalization of Ward's original approach because Ward's linkage as calculated
in~\eqref{eq:ward_kernel} is the variation of within-cluster inertia in the
associated Hilbert space $\mathcal{H}$. This case has been described previously in
\cite{qin_etal_B2003,ahpine_wang_IDA2016}, for instance.

%%%%%%%%%% 

\subsection{Extension to similarity data}
\label{sec:extens-simil-data}

Similarity data also aim at describing pairwise resemblance relations between
the objects of $\Omega$ through a matrix of similarity (or proximity) measures
$S=(s_{ij})_{1\leq i,j \leq n}$. Even though the precise definition of a
similarity matrix can differ within the literature (see
\textit{e.g.}, \cite{hartigan_JASA1967}), it is generally far less constrained 
than kernel matrices. In most cases, the only conditions required to define a 
similarity is the symmetry of the matrix $S$ and the positivity of its diagonal. In some cases, similarity 
measures are also supposed to take non negative values, but we will not make 
this assumption in the present article. 
Since both similarities and kernels describe resemblance relations, it seems 
natural to try to extend the background of Section~\ref{sec:extens-kern-data} to
similarity datasets by using Equation~\eqref{eq:ward_kernel}. This allows the
definition of a linkage, $\delta_S$, between clusters via sums of elements of
$S$. However, this heuristic is not well justified since the quantity
$s_{ii}+s_{jj}-2s_{ij}$ is not necessarily non negative when $S$ is not a
positive definite kernel. Thus, it can not be associated to a squared distance 
as in Equation~\eqref{eq:kernel-squared-distance}.

The previous work of \cite{Miyamoto2015} has explicitly linked similarity and 
kernel data in HAC results. More precisely, for any given similarity $S$, the matrix $S^\lambda = 
(s^\lambda_{ij})_{1\leq i,j\leq n}$ such that $s^\lambda_{ij} := s_{ij} + 
\mathbf{1}_{\{i = j\}} \lambda$ is definite positive for any $\lambda$ larger 
than the absolute value of the smallest eigenvalue of $S$. Therefore, the kernel 
matrix $S^\lambda$ induces a well-defined linkage 
$\delta_{S^\lambda}$ via Equation~(\ref{eq:ward_kernel}), which is linked to 
$\delta_S$ by:
\[
	\delta_{S^\lambda}(G_u,G_v)=\delta_S(G_u,G_v)+\lambda\,.
\]
This proposition justifies the extension of Equation~(\ref{eq:ward_kernel}) to
similarity data with $R_{G_u,G_v}=\sum_{(x_i,x_j) \in G_u\times G_v} s_{ij}$. 
Using this heuristic is indeed equivalent to using a given kernel matrix 
obtained by translating the diagonal of the original similarity $S$: doing so, 
the clustering is unchanged and the linkage values are all translated from 
$+\lambda$ for the kernel matrix, which does not even change the global shape of 
the clustering representation when the heights in this representation are the 
values of the linkage (as discussed in Section~\ref{sec:monoticity}). However, 
as for general dissimilarity data in Section~\ref{sec:extens-diss-data}, the 
interpretation of the linkage as a variation of within-cluster inertia is lost.\\

\paragraph{Conclusion.} In conclusion to this section, there is finally only 
two cases left: the Euclidean case (in which objects are embedded in a direct or 
indirect manner in a Euclidean framework) and the non-Euclidean case. The first 
case includes the standard case, the case of Euclidean distance matrices and 
the case of kernels while the latter case includes general dissimilarity and 
similarity matrices. In the Euclidean case, the original description of the 
Ward's algorithm is valid as such while, in the second, the algorithm can still 
be formally applied in a very similar manner at the cost of a loss of the 
interpretability of the criterion.

\section{Interpretability of dendrograms}
\label{sec:monoticity}

\subsection{Dendrograms}
\label{sec:dendrograms}

The results of HAC algorithms are usually displayed as dendrograms. A dendrogram
is a binary tree in which each node corresponds to a cluster, and, in
particular, the leaves are the original objects to be clustered.  The edges
connect the two clusters (nodes) merged at a given step of the algorithm.  The
height of the leaves is generally supposed to be $h_0 = 0$. In the case of
OCHAC, these leaves are displayed as indicated by the natural ordering of the
objects, while in the general case of unconstrained HAC they are ordered by a
permutation of the class labels that ensures that the successive mergers are
neighbors in the dendrogram. The height of the node corresponding to the cluster
created at merger $t$, $h_t$, is often the value of the linkage. To distinguish 
the height of the dendrogram from the value of the linkage, we will denote by 
$m_t$ the value of the linkage at step $t$. Alternative choices for the values 
of $(h_t)_t$ are discussed in Section~\ref{sec:alt-height}.

Dendrograms are used to obtain clusterings by horizontal cuts of the tree
structure at a chosen height. A desirable property of a dendrogram is thus that
the clusterings induced by such a cut corresponds to those defined by the HAC
algorithm. This property is equivalent to the fact that the sequence of heights
is non-decreasing. When this \emph{monotonicity} property is not satisfied, a
merging step $t$ for which $h_t < h_{t-1}$, is called a
\emph{reversal}. Reversals can be of two types, depending on whether or not they
correspond to a visible \emph{crossover} between branches of the
dendrogram. Mathematically, a crossover corresponds to the case when the height
of a given merger $G_{v_1} \cup G_{v_2}$ is less than the height of $G_{v_1}$ or
the height of $G_{v_2}$. A toy example of reversal with crossover is shown in
Figure~\ref{fig::euclidean-reversal}, between nodes merged at steps 1 and 2, for
the result of OCHAC.
  \begin{figure}
    \centering 
    \begin{tabular}{m{0.5\linewidth} p{0.5\linewidth}}
      \includegraphics[width=\linewidth]{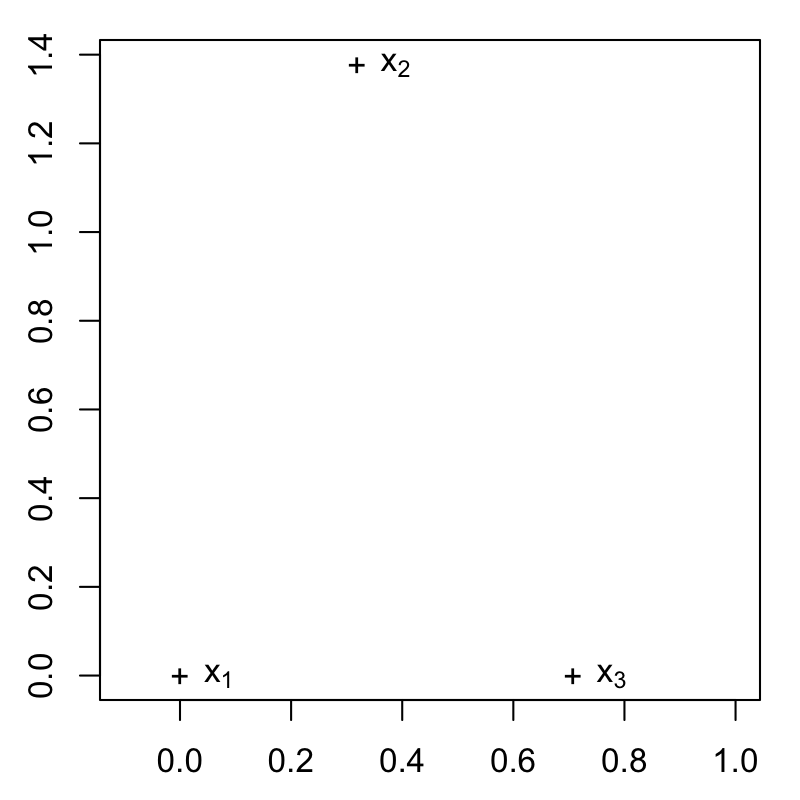} &
      \begin{minipage}{\linewidth}
        $\left(\begin{array}{c}
          x_1\\
          x_2\\
          x_3
        \end{array} \right) = \left(\begin{array}{cc}
          0 & 0\\
          0.45\sqrt{0.5} & \sqrt{1.89875}\\
          \sqrt{0.5} & 0
        \end{array} \right)$\\
        
        $D = \left( \begin{array}{ccc}
          0 & \sqrt{2} & \sqrt{0.5} \\ 
          \sqrt{2} & 0 & \sqrt{2.05} \\ 
          \sqrt{0.5} & \sqrt{2.05} & 0 \\ 
        \end{array}\right)$
      \end{minipage}
    \end{tabular}
    \includegraphics[width=\linewidth]{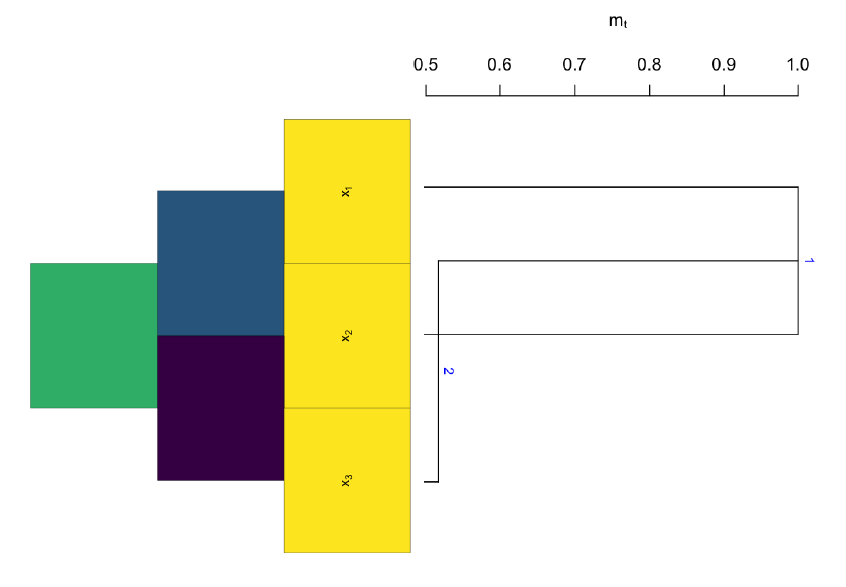}
    \caption{{\bf A crossover for Euclidean OCHAC with height defined as the
        linkage $m_t$.} Top left: Configuration of the objects in
      $\mathbb{R}^2$. Top right : Coordinates of the objects and Euclidean distance matrix corresponding to
      this configuration. Bottom left: Representation of the values of the
      Euclidean distance (dark colors correspond to larger values, so to distant
      objects). Bottom right: Dendrogram obtained from OCHAC (the ordering is
      indicated by the indices of objects) and with the height corresponding to
      Ward's linkage.}
    \label{fig::euclidean-reversal}
  \end{figure}

The goal of this section is to study which settings and which definitions of 
height guarantee the absence of reversals -- with and without crossovers.

\subsection{Monotonicity, crossovers and ultrametricity}

A crossover in a dendrogram automatically implies the non-monotonicity of the
sequence of heights. The converse is true when the height of the dendrogram
corresponds to the value of the linkage (or to a non-decreasing function of the
linkage) for the corresponding merger, by virtue of 
Proposition~\ref{lem:linkage-crossovers} below.
\begin{prop}
  \label{lem:linkage-crossovers}
  Consider a dendrogram whose sequence of heights $(h_t)_t$ is a non-decreasing
  transformation of the linkage values $(m_t)_t$. Then the only reversals that
  can occur are crossovers.
\end{prop}
The proof of Proposition~\ref{lem:linkage-crossovers} is not specific to Ward's 
linkage and is a simple consequence of the
  fact that the linkage is the objective function of the clustering:
\begin{proof}[Proof of Proposition~\ref{lem:linkage-crossovers}]
  Consider an arbitrary merger step of the HAC, characterized by the linkage 
value $m_t$. If the next merger does not involve the newly created cluster,
  then this merger was already a candidate at step $t$. Then, by optimality of
  the linkage value at step $t$, this merger can not be a reversal. Therefore, 
any
  reversal must involve the newly created cluster, and is thus a crossover.
\end{proof}
An important consequence of Proposition~\ref{lem:linkage-crossovers} is that 
when the height of the dendrogram is the corresponding linkage, the absence of 
crossovers is \emph{equivalent} to the monotonicity of the sequence of heights. 

We shall see in Section~\ref{sec:alt-height} that for an arbitrary height, the
absence of crossover in the dendrogram is not necessarily equivalent to the
monotonicity of the sequence of heights. The absence of crossover can be
characterized by a mathematical property of the cophenetic distance associated
to the heights of the dendrogram, called \emph{ultrametricity}. Formally, let us
define, for all $i,j \in \{1, \ldots, n\}$, the \emph{cophenetic distance}
$h_{ij}$ between $i$ and $j$ as the value of the height $h_{t^*}$ such that
$t^*$ is the first step (or the smallest merge number) such that the $i$-th and 
$j$-th objects are in the same cluster. $h$ is said to satisfy the ultrametric 
inequality if:
\begin{align*}
  \forall i,j,k \in \{1,...,n\},\qquad  h_{ij} \leq \max \{ h_{ik},h_{kj} \}.
\end{align*}
As announced, this property is key to ensure the monotonicity of the sequence 
of heights. More precisely, \cite{johnson_P1967} has defined an explicit
bijection between a hierarchy of clusterings with an associated sequence of 
non-decreasing ``heights'' (called ``values'' in the article) and matrix of 
values with a diagonal equal to zero and satisfying the ultrametric inequality. 
It turns out that this bijection explicitly defines the entries of the 
ultrametric matrix as the cophenetic distance of the dendrogram whose heights 
are the one of the associated hierarchy of clusterings. In other words, this 
means that a given sequence of heights defining a dendrogram is non-decreasing 
if and only if the cophenetic distance associated to this dendrogram (or 
equivalently to this sequence of heights) satisfies the ultrametric inequality.

\subsection{Monotonicity of Ward's linkage}

Ward's linkage corresponds to the variation of within-cluster inertia, so that
the monotonicity of the linkage is ensured for Ward's standard HAC
algorithm with Euclidean data. More generally, \cite{Batagelj1981} gives
necessary and sufficient conditions based only on the Lance-Williams
coefficients that ensures monotonicity for a given linkage. These results apply
to the extensions of HAC to non Euclidean datasets and show that the 
monotonicity 
of the linkage values is always ensured for standard HAC with Ward's linkage. In
addition, \cite{Ferligoj1982} give necessary and sufficient conditions on the
Lance-Williams coefficients to ensure the monotonicity of the linkage values in
constrained HAC, for an arbitrary symmetric relational constraint. These
conditions are not fulfilled for Ward's linkage. Therefore, monotonicity is not
guaranteed for CCHAC with Ward's linkage, as also noted by \cite{Grimm1987} for
the specific case of OCHAC. It can be shown that even for Euclidean data, the
contiguity constraint can induce non increasing linkage values for some steps of
the algorithm, as illustrated by Figure~\ref{fig::euclidean-reversal}. 

More precisely, if we consider OCHAC, the following 
proposition establishes necessary and sufficient conditions on a dissimilarity 
$d$ to observe a reversal at a given step of OCHAC when the height is defined by 
Ward's linkage:

\begin{prop}
  \label{prop::decreasing}
  Suppose that $\Omega = \{x_i\}_{i=1, \ldots, n}$ is equipped with the 
symmetric contiguity relation $x_i \mathcal{R} x_j \Leftrightarrow \vert i-j 
\vert=1$ (OCHAC). Denote by $l$ and $r$ the indices of the left and right 
clusters merged at a given step $t$, and by $\bar{l}$ and $\bar{r}$ their own 
left and right cluster, respectively.  Then there is a reversal at step $t+1$ 
for the height defined by the linkage if and only if:
\begin{align}
\label{eq:decreasing}  
  \delta(G_l, G_r) \geq \min\left(
\frac{g_{\bar{l}l} \delta(G_{\bar{l}}, G_{l}) + g_{\bar{l}r} \delta 
(G_{\bar{l}}, G_{r})}{g_{\bar{l}l} + g_{\bar{l}r}}
,
\frac{g_{l\bar{r}} \delta(G_l, G_{\bar{r}}) + g_{r\bar{r}} \delta (G_r, 
G_{\bar{r}})}{g_{l\bar{r}} + g_{r\bar{r}}}
\right)\,
\end{align}
where we have used the notation $g_{uv} := |G_u \cup G_v| = |G_u| + |G_v|$.
\end{prop}

The fact that Condition~\eqref{eq:decreasing} involves clusters contiguous to 
the last merger is a consequence of Proposition~\ref{lem:linkage-crossovers}. 
The formulation of Condition~\eqref{eq:decreasing} is quite intuitive: 
crossovers correspond to situations in which the Ward linkage between two newly 
merged clusters is larger than a (weighted) average Ward linkage between each of 
these two clusters and one of the contiguous clusters. The proof of 
Proposition~\ref{prop::decreasing} is given in 
Appendix~\ref{sec:proof-decreasing}.

Let us apply Proposition \ref{prop::decreasing} to the specific case of the 
first and second mergers in the algorithm.  Assuming that the optimal merger at 
step 1 is between the $l$-th and $r$-th objects, and recalling that the Ward 
linkage between two singletons is simply $\delta(\{u\},\{v\}) = d^2_{uv}/2$, 
Condition~\eqref{eq:decreasing} reduces to:
$$ 2d_{l, r}^2 > \min\left(d^2_{\bar{l}, l} + d^2_{\bar{l}, r}, d^2_{r, \bar{r}} 
+ d^2_{l, \bar{r}}\right)$$
In particular, given the $p-1$ distances $(d^2_{i,i+1})_{1 \leq i \leq p-1}$ 
that determine the first step of the OCHAC algorithm, it is always possible to 
find an adversarial dissimilarity yielding a reversal at the second step, 
\textit{e.g.}, by choosing $d_{l, \bar{r}}$ such that $d^2_{l, \bar{r}} < 
2d^2_{l, r} - d^2_{r, \bar{r}}$. This is the case in the counter-example of 
Figure~\ref{fig::euclidean-reversal}.

\paragraph{An example of relevant  reversal for OCHAC.}
Because of the possible presence of crossovers in OCHAC even in a simple
Euclidean setting, CCHAC may appear as a deteriorated version of standard HAC,
where the optimal merger is chosen within a reduced set of possible mergers
compared to the unconstrained version. One may then expect that the total
within-cluster inertia at a given step of the algorithm is larger than for the
unconstrained version that chooses the ``optimal'' merger at this step (that is,
the merger with the smallest increase of the total within-inertia). In addition,
the algorithm does not necessarily exhibit a clear and understandable monotonic
evolution of the objective criterion, $(m_t)_t$. However, it can be shown, even
in a very simple example, that OCHAC can lead to better solutions in terms of
within-cluser intertia, when the constraint is consistent to the spatial
structure of the data. This fact is illustrated in
Figure~\ref{fig::outperforms} (detailed analysis of all examples 
and
  counter-examples of this section is provided in
  Appendix~\ref{sec:counterexamples_standardHAC}). In this example, 7 data
points are displayed in $\mathbb{R}^2$ with an order constraint illustrated by a
line linking two points allowed to be merged. In this situation, $(m_t)_t$ is
indeed non monotonic for OCHAC (bottom left figure) but leads to a better total
within-cluster inertia for $k=3$ clusters (vertical green line), which is also
more relevant for the data configuration (top figures). This is a typical case
where the constraint forces the algorithm to explore under-efficient
configurations but that can be aggregated into a better solution, contrary to
the unconstrained algorithm. This is explained by the fact that even the 
unconstrained algorithm is greedy, by construction, and thus not optimal 
compared to an exhaustive search of the best partition in $k$ classes.
\begin{figure}[ht]
  \centering
  \begin{tabular}{p{0.5\linewidth} p{0.5\linewidth}}
  \includegraphics[width=0.95\linewidth]{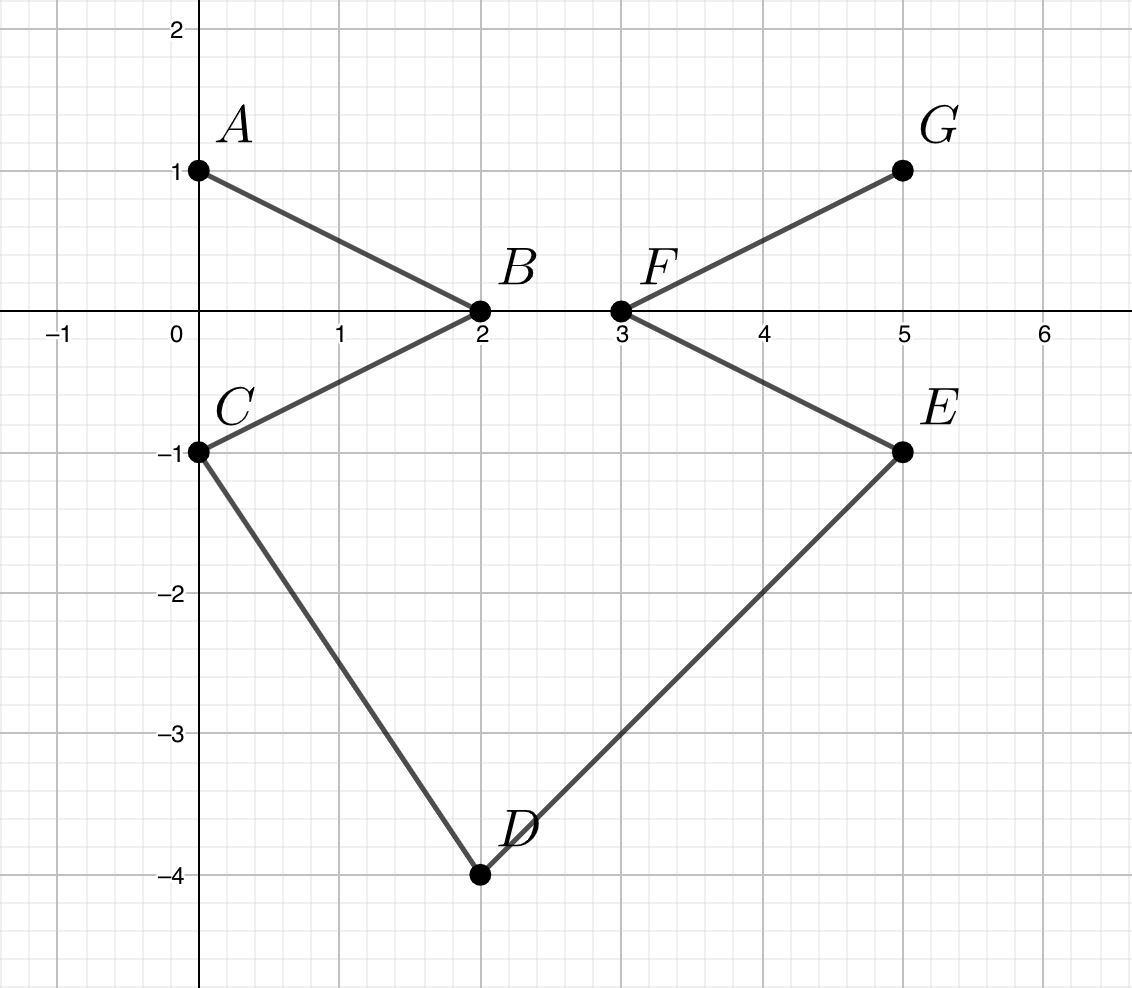} &
  \includegraphics[width=0.95\linewidth]{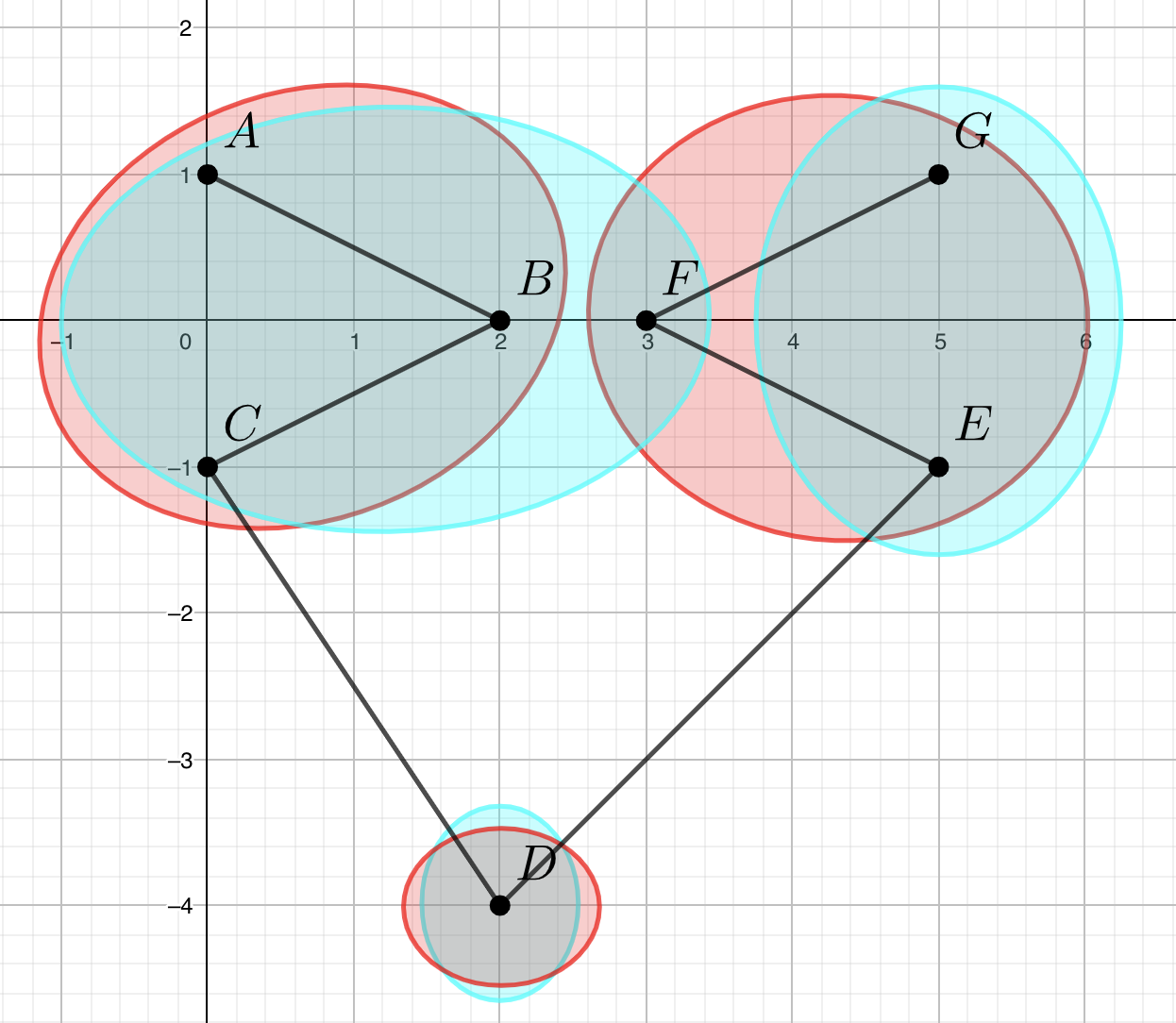}
   \end{tabular}
  \includegraphics[width=0.9\linewidth]{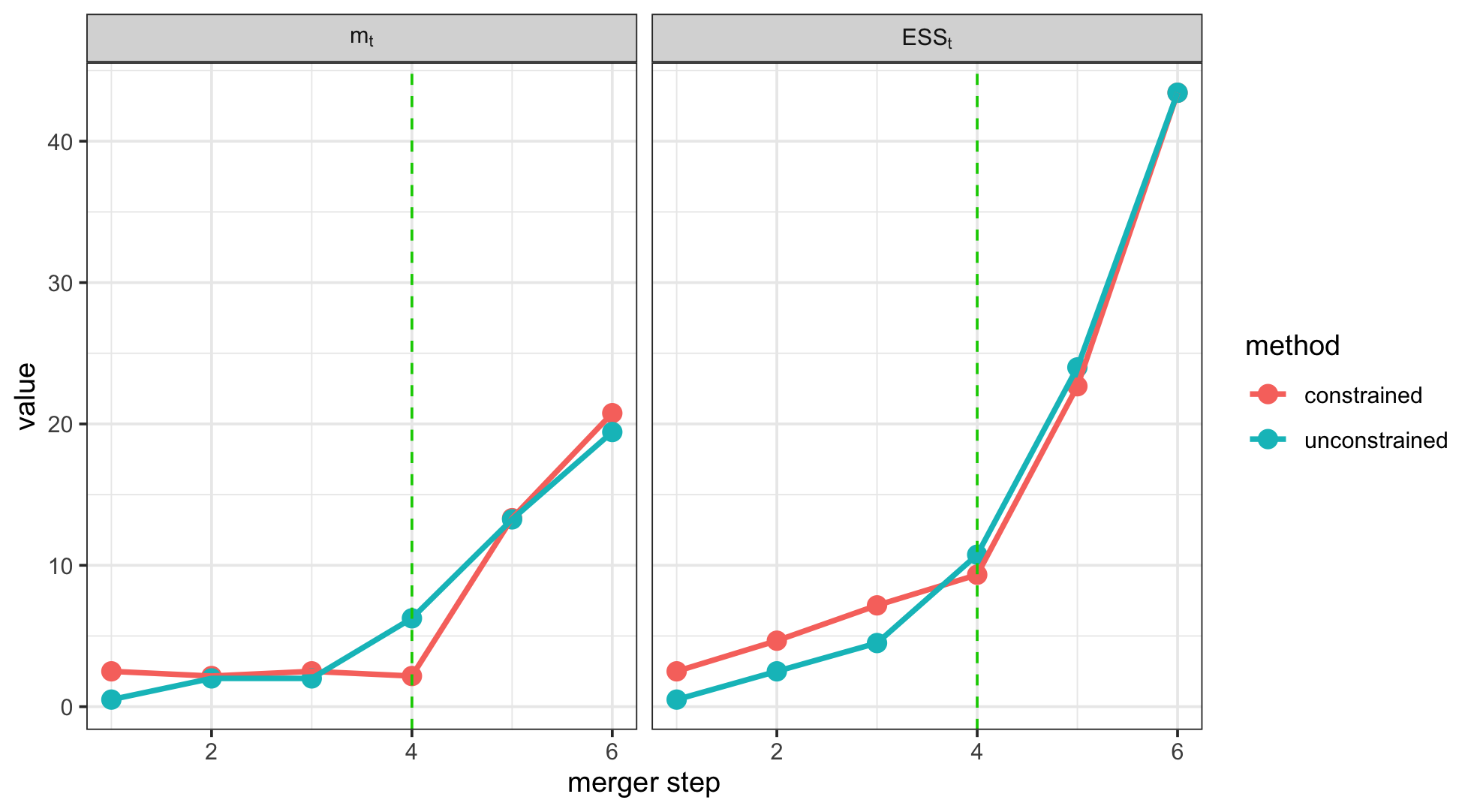}
  \caption{\textbf{Simple configuration in which OCHAC outperforms standard 
HAC.} Top left: Initial configuration with the order constraint represented by 
straight lines. Top right: Clustering with 3 clusters as 
produced by OCHAC (red) and standard HAC (blue). Bottom: Evolution of $(m_t)_t$ 
and of the 
total within-cluster inertia (also called, Error Sum of Squares: (ESS$_t)_t$) 
along the clustering processes, the green line correspond to the 3 components 
clustering.}
  \label{fig::outperforms}
\end{figure}

\clearpage
\subsection{Monotonicity of alternative  heights}
\label{sec:alt-height}

Since reversals can occur in CCHAC dendrograms with Ward's linkage, alternative 
definitions of the height have been proposed to improve the interpretability 
of the result in this case. They are defined as quantities related to the 
heterogeneity of the partition. In this section, we study the monotonicity of 
such 
alternative heights.

\cite{Grimm1987} presents three alternative heights to the standard 
\emph{variation of within-cluster inertia} ($m_t$):
\begin{itemize}
  
	\item the \emph{within-cluster (pseudo-)inertia} (or \emph{Error Sum of 
Squares}) that corresponds to the value of the objective function. In this 
case, 
the height at step $t$ is given by: 
	\[
		\mbox{ESS}_t = \sum_{u=1}^{n-t} I(G_u^{t+1}),
	\]
	where $\mathcal{P}^{t+1} = \{ G_u^{t+1} \}_{u=1,\ldots,n-t}$ is the partition 
obtained at step $t$ of the algorithm. This alternative height is very natural 
(and the one implemented in the \RR{} package \pkg{rioja} \citep{juggins_2012} for OCHAC) since it 
corresponds to the criterion whose minimization is approximated by HAC (and 
OCHAC) in a greedy way;
	\item the \emph{(pseudo-)inertia of the current merger}, which is defined as:
	\[
    I_t = I(G_u^t \cup G_v^t)
	\]
  where $G_u^t$ and $G_v^t$ are the two clusters merged at step $t$. 
\cite{Grimm1987} remarks that this 
measure is very sensitive to the cluster size $|G_u^t| + |G_v^t|$. 

\item the \emph{average (pseudo-)inertia of the current merger}, that has been
  designed so as to avoid the bias related to the cluster size in $I_t$. It is
  defined as:
  \[
    \overline{I}_t = \frac{I_t}{|G_u^t| + |G_v^t|}
  \]
\end{itemize}

\paragraph{Standard HAC: Known properties of alternative heights.} Note that 
$\mbox{ESS}_t = \sum_{t' \leq t} m_{t'}$. As explained in 
Section~\ref{sec:monoticity}, 
$(m_t)_t$ is monotonic for standard HAC, both for Euclidean and non Euclidean 
data. Since $m_0 = 0$ by definition, this ensures the monotonicity of 
$(\mbox{ESS}_t)_t$, for Euclidean and non Euclidean data in the case of 
standard 
HAC. 

On the contrary, $I_t$ and $\overline{I}_t$ may induce reversals even for 
standard HAC and Euclidean data. More importantly, contrary to 
the case when the height of the dendrogram is $m_t$, even when the ultrametric 
property is satisfied, the monotonicity is not ensured for these criteria. 
This is illustrated in Figure~\ref{fig::I} (and in Figure~\ref{fig::AI} of the 
Appendix~\ref{sec::counter_barI}), for $I_t$ (and for $\bar{I}_t$, 
respectively) 
and data in $\mathbb{R}^2$. 
\begin{figure}
  \centering 
  \begin{tabular}{m{0.5\linewidth} p{0.5\linewidth}}
    \includegraphics[width=\linewidth]{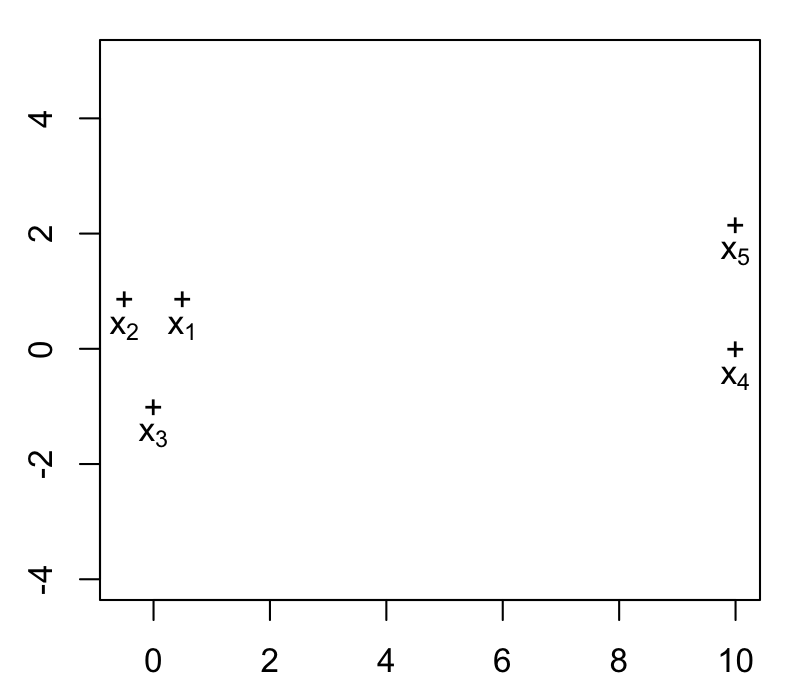} &
    \begin{minipage}{\linewidth}
    $\left(\begin{array}{cc}
    x_1\\
    x_2\\
    x_3\\
    x_4\\
    x_5
  \end{array} \right) =
  \left(\begin{array}{cc}
    \frac{1}{2} & \frac{\sqrt{3}}{2} \\
    -\frac{1}{2} & \frac{\sqrt{3}}{2} \\
    0 & -1\\
    10 & 0\\
    10 & 2.16
  \end{array} \right)$\\
  
$D \approx
\left(\begin{array}{ccccc}
0.00 & 1.00 & 1.93 & 9.54 & 9.59 \\ 
1.00 & 0.00 & 1.93 & 10.54 & 10.58 \\ 
1.93 & 1.93 & 0.00 & 10.05 & 10.49 \\ 
9.54 & 10.54 & 10.05 & 0.00 & 2.16 \\ 
9.59 & 10.58 & 10.49 & 2.16 & 0.00 
\end{array}\right)$
\end{minipage}
  \end{tabular}
  \includegraphics[width=\linewidth]{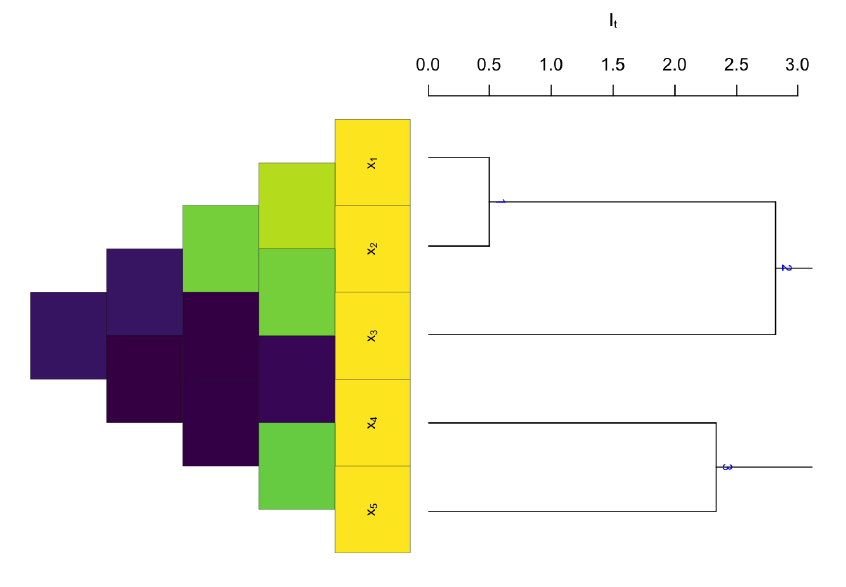}
  \caption{{\bf A reversal for Euclidean standard HAC with height defined as 
$I_t$.} Top left:
    Configuration of the objects in $\mathbb{R}^2$. Top right: Coordinates of 
the objects and Euclidean distance matrix corresponding to
      this configuration. Bottom left: Representation
    of the values of the dissimilarity (dark colors correspond to larger values,
    so to distant objects). Bottom right: dendrogram obtained from standard
    HAC. Only the first 3 merges of the dendrogram is represented to ensure a
    comprehensive view of the sequence of heights.}
  \label{fig::I}
\end{figure}

  In this case, the dendrogram has a conventional look but the mergers are not 
ordered by increasing heights. For instance, in Figure~\ref{fig::I}, the cluster 
merged at step 2 
is above the one at step 3. Hence, cutting the dendrogram at height $h = 2.5$ 
leads to a clustering into $\{x_1, x_2\}, \{x_3\}, \{x_4, x_5\}$, but this 
clustering does not belong to the sequence of clusterings induced by the HAC 
(where the clustering in 3 clusters is  the one obtained after the second 
merger, that is, $\{x_1, x_2, x_3\}, \{x_4\}, \{x_5\}$).

\paragraph{CCHAC: Known properties of alternative heights.}
Figures~\ref{fig::I} and \ref{fig::AI} (the latter in
Appendix~\ref{sec::counter_barI}) provide counter-examples for the monotonicity 
of
$(I_t)_t$ and $(\bar{I}_t)_t$ in the Euclidean case for HAC. If the objects are
pre-ordered as the nodes in these figures, then OCHAC and standard HAC give 
identical
hierarchical clusterings. Therefore, these examples also provide
counter-examples for the monotonicity of $(I_t)_t$ and $(\bar{I}_t)_t$ in the
Euclidean case for OCHAC, and show that there is no guarantee for monotonicity 
in the case of general CCHAC. The fact that $(\bar{I}_t)_t$ is not necessarily
monotonous for OCHAC has already been mentioned by \cite{Grimm1987}. 

\paragraph{CCHAC: Within-cluster pseudo-inertia for dissimilarity data.} The
only unanswered case is whether $(\mbox{ESS}_t)_t$ is monotonic or not for CCHAC
and non Euclidean data. We provide a counter-example that proves that the
monotonicity is not ensured in this case: Figure~\ref{fig::ESS} shows that the
dendrogram obtained from OCHAC on a given non-Euclidean dissimilarity $D$ 
contains a crossover $(m_4 < m_3)$. In particular, the associated sequence of 
heights is not monotonic. However, Proposition~\ref{lem:linkage-crossovers} 
ensures that $(\mbox{ESS}_t)_t$ has the nice property that the absence of 
crossovers is equivalent to its monotonicity. Indeed, as $(\mbox{ESS}_t)_t$ 
corresponds to the cumulative sums of the linkage $(m_t)_t$, the mapping 
between $m_t$ and ESS$_t$ is equal to the addition of ESS$_{t-1}$. As, by 
definition, ESS$_{t-1}$ is, as any $I(G_u^{t-1})$, positive, this ensures that 
this mapping is non-decreasing.

\begin{figure}
\centering 
\begin{tabular}{m{0.5\linewidth}}
\begin{minipage}{\linewidth}
$D=\left(\begin{array}{cccccc}
 0 & \sqrt{1.99} & \sqrt{1.99} & \sqrt{1.99} & 0.1 & 1 \\
 \sqrt{1.99} & 0 & \sqrt{2} & \sqrt{1.99} & 0.1 & 1 \\
 \sqrt{1.99} & \sqrt{2} & 0 & \sqrt{2} & 0.1 & 1 \\
 \sqrt{1.99} & \sqrt{1.99} & \sqrt{2} & 0 & \sqrt{2} & 1 \\
0.1 & 0.1 & 0.1& \sqrt{2} & 0 & \sqrt{2} \\
 1 & 1 & 1 & 1 & \sqrt{2} & 0 
\end{array}\right)$
\end{minipage}
\end{tabular}
\includegraphics[width=\linewidth]{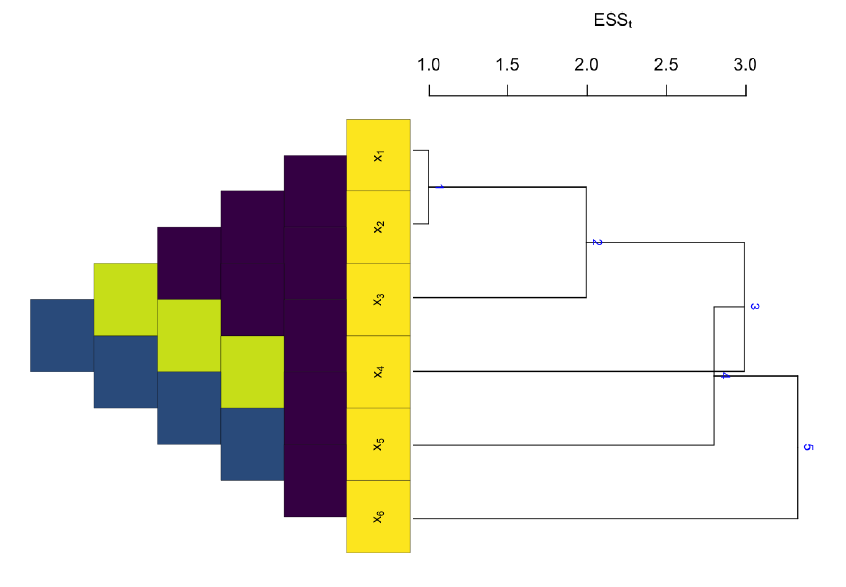}
\caption{{\bf A reversal for non-Euclidean OCHAC with height defined as 
$\mbox{ESS}_t$}. Top: Dissimilarity matrix. Bottom left: Representation of the 
values of the 
dissimilarity $D$ (dark colors correspond to larger values, so to distant 
objects). Bottom right: Dendrogram obtained from OCHAC (the ordering is 
indicated by the indices of objects) and with the height corresponding to 
ESS$_t$. }
\label{fig::ESS}
\end{figure}
Table~\ref{table::summary} summarizes the properties
of the different types of heights, respectively for standard HAC and CCHAC. Note 
that the monotonicity of $\textrm{ESS}_t$ is a consequence of the positivity of 
$m_t$.

\begin{table}[ht]
  \centering
  {\footnotesize \begin{tabular}{c c|l|l|l|l}
    \toprule
    &
    & $m_t$
    & $\textrm{ESS}_t$
    & $I_t$
    & $\overline{I}_t$ \\ \midrule 
    {\normalsize HAC} & Euclidean & \ding{51}\cite{ward_JASA1963} &  
\ding{51}\cite{ward_JASA1963} & \ding{53} [Fig.~\ref{fig::I}] &  
\ding{53}[Fig.~\ref{fig::AI}] \\
    &Non Euclidean & \ding{51}\cite{Batagelj1981} &  
\ding{51}\cite{Batagelj1981}   & \ding{53}  [Fig.~\ref{fig::I}] &  
\ding{53}[Fig.~\ref{fig::AI}]\\
    \midrule
    {\normalsize CCHAC} & Euclidiean & \ding{53}\cite{Grimm1987} &  
\ding{51}\cite{Grimm1987}  & \ding{53} [Fig.~\ref{fig::I}] & 
\ding{53}\cite{Grimm1987} \\
    &Non Euclidean & \ding{53}\cite{Grimm1987} &  \ding{53} 
[Fig.~\ref{fig::ESS}]  & \ding{53} [Fig.~\ref{fig::I}] &  
\ding{53}\cite{Grimm1987}\\
    \bottomrule
  \end{tabular}}\\
  \caption{Monotonicity of heights for standard HAC (top) and CCHAC 
(bottom).}
  \label{table::summary}
\end{table}

\clearpage
\section{Simulation}
\label{sec:simulations}

HAC can be seen as a greedy algorithm to solve the problem of finding the
partition with minimal within-cluster inertia ESS$_t$ of $n$ objects into $n-t$
classes, for each $t = 1 \dots n-1$. It may be expected that the inertia of the
partitions will be lower for HAC than OCHAC, since the possible mergers in OCHAC
are chosen among a subset of the possible mergers in HAC. Can we
quantify the impact of the order constraint on the quality of the partitions (as
measured by ESS) obtained for HAC and OCHAC, depending on the strength of the
actual order structure in the data? In this section, we address this question by
analyzing Hi-C data \citep{dixon_etal_N2012}, which present a strong order
structure, as illustrated by Figure~\ref{fig::hicmap}. We use a perturbation
process to progressively break the consistency between the data structure and
the constraint imposed in OCHAC. 
% In this section, we first start by describing 
% the simulation process. Then, we comment on the way standard HAC and OCHAC 
% results compare. Finally, we focus the comparison on a description of the 
% reversals for the different heights and the two versions (constrained and 
% unconstrained) of the algorithm.

\subsection{Data and method}

Hi-C studies aim at characterizing proximity relationships in the 3D structure
of a genome, by measuring the frequency of physical interaction between pairs of
genomic locations via sequencing experiments. Formally, an Hi-C map is a
symmetric matrix $S=(s_{ij})_{i,j}$ in which each entry $s_{ij}$ is equal to the
frequency of interaction between genomic loci $i$ and $j$. Here, a locus is a
fixed-size interval of genomic positions, also called a ``bin''. Hi-C maps are
classically represented by the upper triangular part of the matrix, as shown in
Figure~\ref{fig::hicmap}. The matrix has a strong diagonal structure that
reflects the linear order of DNA within chromosomes (loci that are close along
the genome are more frequently interacting than distant loci). An important
question in Hi-C studies is to identify Topologically Associating Domains
(TADs), which are self-interacting genomic regions appearing to be more compact
than the rest of the genome. Indeed, TADs have been shown to play an important
role in gene regulation \citep{dixon_etal_N2012}. A number of TAD detection
methods have been proposed (see \textit{e.g.}, \cite{zufferey2018comparison} for
a review) and some are based on HAC or OCHAC
\citep{fraser_etal_MSB2015,haddad_etal_NAR2017,ambroise_etal_p2018}. This is
both natural, since Hi-C maps can be seen as similarity matrices, and formally
justified, as explained in Section~\ref{sec:extens-simil-data}.

The simulations in this section are based on a single chromosome (chromosome 3)
from an experiment in human embryonic stem cells (hESC;
\cite{dixon_etal_N2012}\footnote{The pre-processed and normalized data have been
  downloaded from the authors' website at
  \url{http://chromosome.sdsc.edu/mouse/hi-c/download.html} (raw sequence data
  are also published on the GEO website, accession number GSE35156).}). The
downloaded Hi-C matrix contains 4,864 bins. It has been obtained with a
bin size of 40kb and normalized using ICE \citep{imakaev_etal_NM2012}. We
further performed a log-transformation of the entries to reduce the distribution
skewness prior clustering.

\begin{figure}
\centering 
\includegraphics[width=0.5\linewidth]{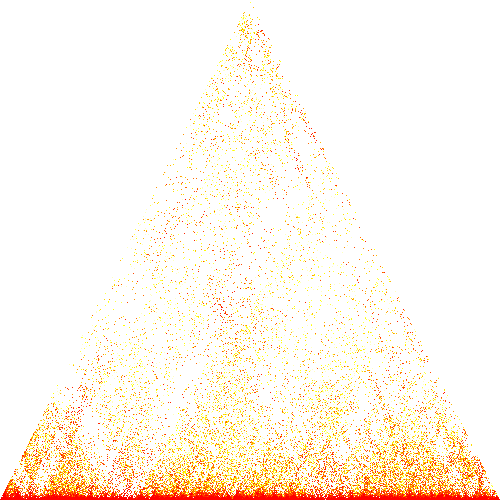}
\caption{{\bf Graphical representation of an Hi-C map}. Horizontal position is defined by the indices of bins within a single chromosome. Intensities of the frequency of physical interaction between bins are represented by levels of red.}
\label{fig::hicmap}
\end{figure}

In order to assess the influence of the data structure on the quality of the
partitions obtained by OCHAC and standard HAC algorithms, we have used a
perturbation process to progressively remove the strong diagonal in the original
Hi-C map. The perturbation consists in swapping two entries, $s_{ij}$
and $s_{i'j'}$ of the matrix, in which $(i,j)$ and $(i',j')$ have been randomly
sampled with uniform probability among the pairs $\{(u,v),\, (u',v')\}$ for
which $u \leq v$, $u' \leq v'$ and $s_{uv}+s_{u'v'} > 0$, where the last
condition avoids swapping entries that are both zero. The proportion of such
swapped pairs, which we call perturbation level, varied from 0\% up to 90\%
(Figure~\ref{fig::perturbation}).

This process was repeated 50 times to allow assessing the variability. Since 
obtained matrices are not necessarily positive definite, we translate their 
diagonal by a small quantity that ensures the positivity of all 
$d_{ij}^2=s_{ii}+s_{jj}-2s_{ij}$ as described in 
Section~\ref{sec:extens-kern-data}.

\begin{figure}
  \centering 
  \includegraphics[width=\linewidth]{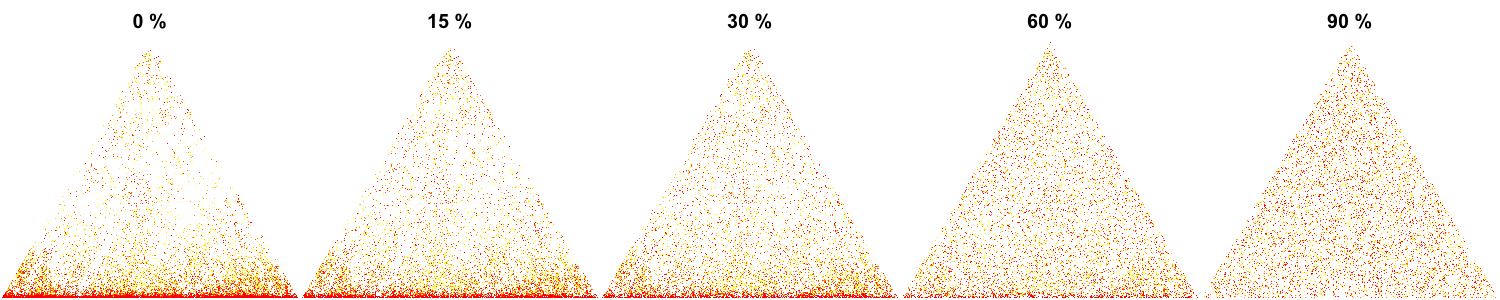}
  \caption{{\bf Illustration of the perturbation process}. From left to right : 
example of Hi-C maps corresponding to increasing perturbation levels.}
\label{fig::perturbation}
\end{figure}

All simulations were performed with \RR{}. The results for standard HAC were 
computed with the function \texttt{hclust} (from the \pkg{stats} package) and 
those for OCHAC were computed with the function \texttt{adjClust} (from the 
\pkg{adjclust} package). Figures were obtained using \pkg{adjClust} or 
\pkg{ggplot2} \citep{wickham_gEGDA2016}.

\subsection{Comparison of standard HAC and OCHAC results}

In this section, the results of standard HAC and OCHAC are compared through the 
corresponding height sequences the dendrograms, and through clusterings obtained 
by horizontal cuts of the dendrograms. While the first two are a direct output 
of the HAC process, clusterings are obtained using a model selection strategy. 
We have considered two such strategies: the broken stick  
\citep{bennett_NP1996}, as implemented in \pkg{adjclust}, and the slope 
heuristic \citep{arlot_etal_capushe2016}, as implemented in \pkg{capushe}. As 
both of them gave similar results, we chose to report here only the results 
obtained for the broken stick heuristic.

\paragraph{Height sequences.} Figure~\ref{fig::heights_adjhcl} shows the
evolution of $m_t$ (normalized by its maximal value among both methods at a
given permutation level) and ESS$_t$ (normalized by the total inertia of the set
of bins) along the two clustering processes for increasing perturbation
levels. For the original dataset, which presents an organization strongly
consistent with the order constraint, the heights of standard HAC and OCHAC are
very similar. However, interestingly, OCHAC improves the objective criteria
(ESS$_t$ and $m_t$) for low perturbation levels (15\%-30\%) across a wide range of merging levels. 

More specifically, we compared the heights obtained for HAC and OCHAC at the
merger number selected by the broken stick heuristic (\citet{bennett_NP1996};
vertical lines in Figure~\ref{fig::heights_adjhcl}). At these numbers of
clusters or in their close neighbourhood, ESS$_t$ is always smaller for OCHAC,
which we interpret as more homogeneous clusterings for OCHAC than for HAC. The
magnitude of the improvement achieved by OCHAC with respect to HAC depends on
the perturbation level: for the original data, it is close to 5\%, whereas it is
much larger (25-30\%) when the perturbation level is 15\%-30\%. It then
decreases again ($<20$\%) for larger perturbation levels (60\%).

The fact that OCHAC can achieve lower values than HAC for ESS$_t$ and $m_t$  may
be counter-intuitive, since --as explained at the beginning of Section
\ref{sec:simulations}-- possible mergers in OCHAC are chosen among only a subset
of the possible mergers in standard HAC. In fact, HAC itself is a heuristic for
the minimization of ESS$_t$, because of its hierarchical agglomerative nature;
in contrast, the optimal clustering at step $t$ in the sense of ESS$_t$ may not
necessary be obtained by merging two clusters of the optimal clustering at step
$t-1$. This result illustrates the robustness to noise of the constrained
approach, which is very interesting in practice: in Hi-C experiments, for
instance, many biases (genomic, experimental, etc.) are encountered. Thus, OCHAC
has to be preferred in such contexts and will additionally result in a lower
computational cost.
\begin{figure}[ht]
  \centering 
  \includegraphics[width=\linewidth]{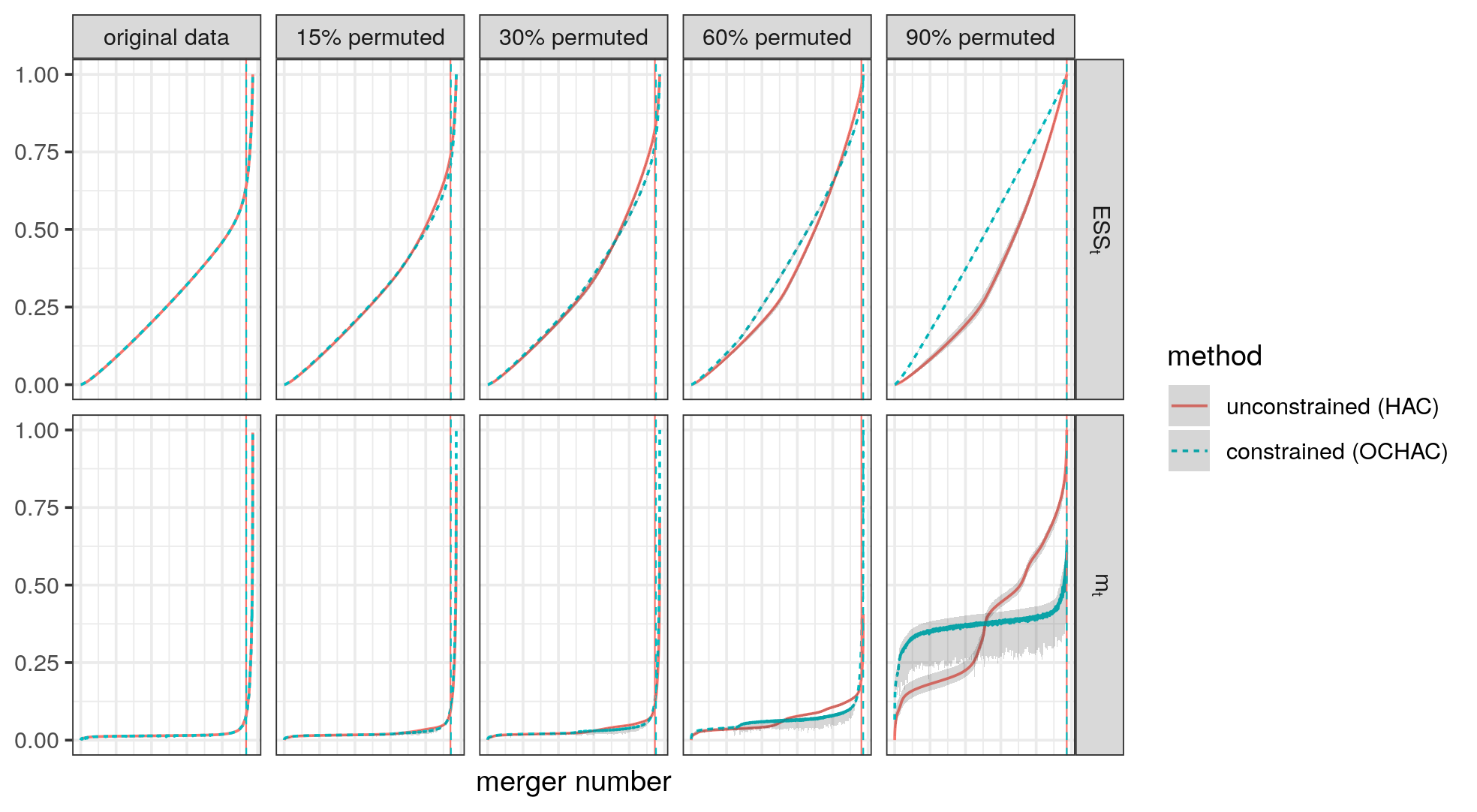}
  \caption{{\bf Comparison of the height sequences for standard HAC (red, solid) 
and OCHAC (blue, dashed)} for ESS$_t$ (top) and $m_t$ (bottom) with various 
levels of perturbations of the original Hi-C matrix. The curves correspond to 
the average criteria over 50 simulations and the grey shadows correspond to the 
minimum and maximum of the criteria over 50 simulations. The vertical lines 
correspond to the average number of clusters chosen by the broken stick 
heuristic, respectively for standard HAC and OCHAC (red, solid and blue, 
dashed).}
  \label{fig::heights_adjhcl}
\end{figure}

For perturbation levels larger than 60\%, the data structure is no more
compatible with the constraint (see Figure~\ref{fig::perturbation}) and standard
HAC seems to performs globally better than OCHAC, as expected. In addition, in
this extreme situation, OCHAC exibits very large reversals in $m_t$ values (seen
with the grey shadow in Figure~\ref{fig::heights_adjhcl}), that are due to
sudden breaks in the quality of the clusterings, induced by the constraint. The
presence of such large reversals is a practical and visible indication that the
constraint is not relevant for the data and that OCHAC should not be used.

\paragraph{Dendrograms and clusterings.}
The same type of conclusion can be drawn when comparing not just the heights 
of the dendrograms but the dendrograms themselves or the clusterings induced by 
these dendrograms. Figure~\ref{fig::distcoph_baker_adjhcl} shows the 
distribution of Baker's $\gamma$ coefficients \citep{baker_JASA1974} between 
the dendrograms of standard HAC and OCHAC versus the perturbation level. This coefficient corresponds to the Spearman 
correlation between the pairwise values of the cophenetic distances between 
pairs of objects as induced by the dendrogram. 
\begin{figure}[ht]
  \centering 
  \includegraphics[width=\linewidth]{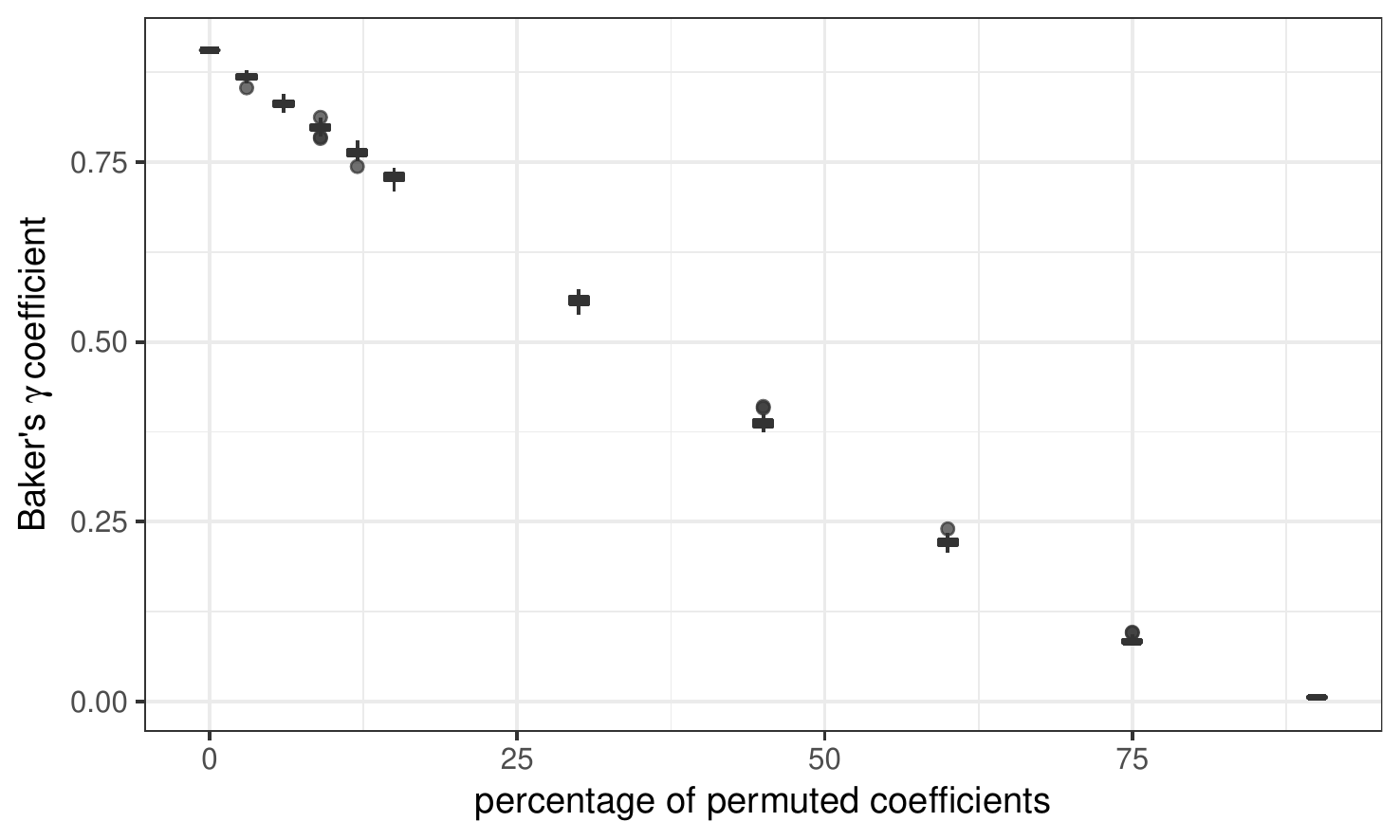}
  \caption{{\bf Baker's $\gamma$ coefficient between OCHAC and standard 
HAC dendrograms versus perturbation level.}}
  \label{fig::distcoph_baker_adjhcl}
\end{figure}
As the perturbation level increases, Baker's $\gamma$ linearly decreases from a value close to 1 (implying very similar dendrograms) to a value close to 0 (implying completely different dendrograms). 

Finally, we compared the clusterings obtained by the broken stick heuristic \citep{bennett_NP1996} as follows. For larger perturbation levels (more
than 60\%) of permuted coefficients, we obtained a trivial clustering with only
one cluster, a strong indication that the cluster structure had disappeared at
these levels. For lower perturbation levels, the obtained clusterings were
compared using the Normalized Mutual Information (NMI,
\cite{danon_etal_JSM2005}). As for Baker's $\gamma$, the NMI values obtained for
the original data and low levels of perturbations (up to 30\%) are very close to
1, which shows a strong similarity of the induced clusterings. As the
perturbation level increases, the obtained partitions became more and more
different, with NMI values below 0.6 (results not shown).

\subsection{Reversals for the different heights}

In this section, we investigate the reversals obtained for different heights and
for standard HAC and OCHAC. Figure~\ref{fig::reversals} gives the evolution of
the percentage of reversals (relative to the total number of simulations, 50),
for standard HAC and OCHAC and for the different types of heights, along the
hierarchical clustering process.
\begin{figure}[ht]
  \centering 
  \includegraphics[width=\linewidth]{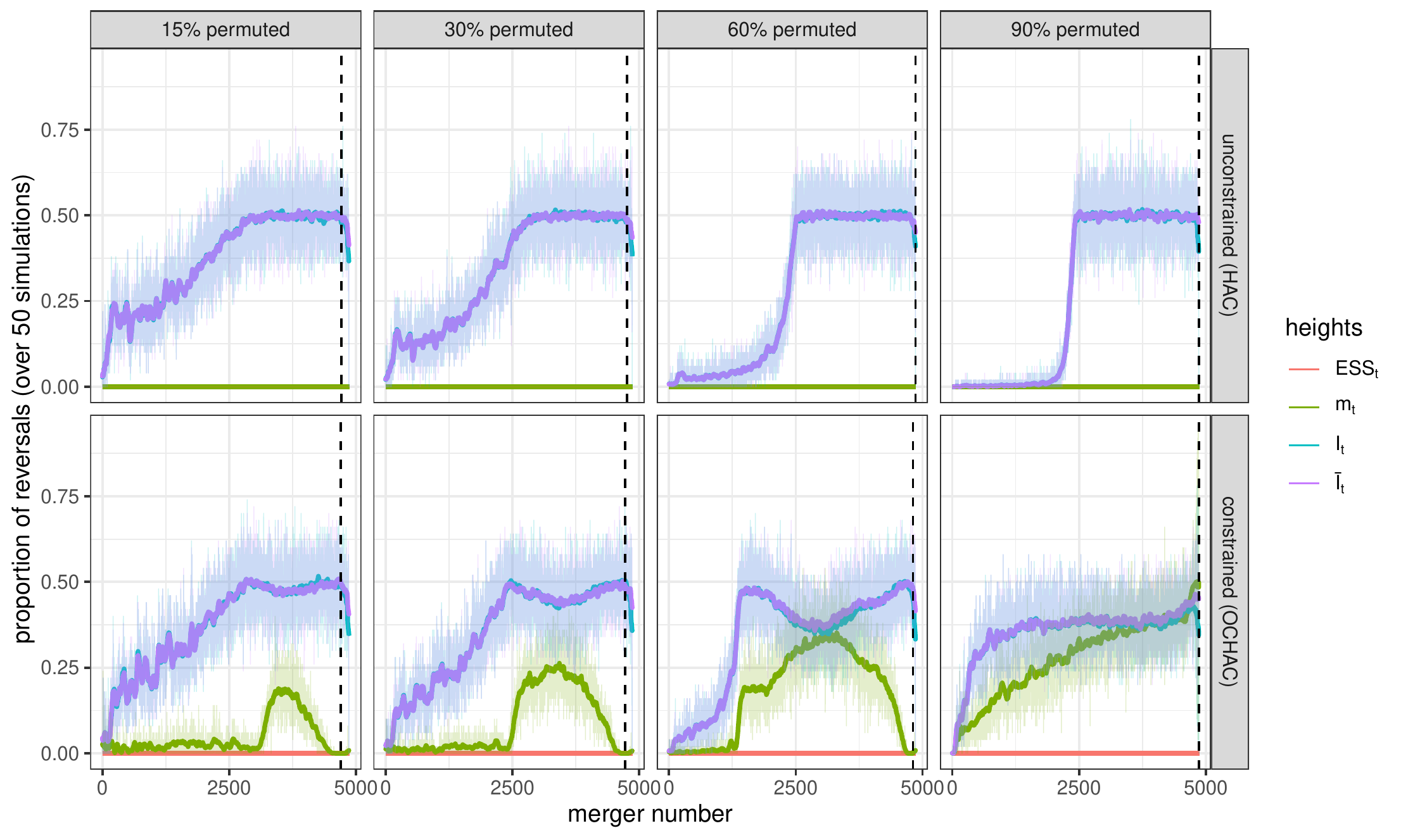}
  \caption{{\bf Evolution of the number of reversals} for ESS$_t$, $m_t$, $I_t$ 
and $\bar{I}_t$ for standard HAC (top) and OCHAC (bottom) for increasing levels
    of perturbation of the original Hi-C matrix. The background shadow is the 
actual value and the strong line is a smoothed value (box kernel, bandwith 
equal to 50).   
    The dotted vertical line
    corresponds to the average number of clusters chosen by the broken stick
    heuristic.}
  \label{fig::reversals}
\end{figure}

As expected from
Section~\ref{sec:HAC_extensions} (Table~\ref{table::summary}),
$(\mbox{ESS}_t)_t$ does not have reversals and $(m_t)_t$ only has reversals for
OCHAC. When the perturbation level increases, the evolution of the number of
reversals in $(I_t)_t$ and $(\bar{I}_t)_t$ is markedly different from that of
$(m_t)_t$. For the smallest perturbation levels (up to 30\%), the number of
reversals of $(m_t)_t$ is close to 0, while it ranges from 10 to 50\% for
$(I_t)_t$ and $(\bar{I}_t)_t$. At these perturbation levels, $(m_t)_t$ almost
never has a reversal at a merger number that
corresponds to the number of clusters chosen by the broken stick heuristic: most
reversals are concentrated at a merger number smaller than the merger chosen by
the broken stick heuristic.  Actually, for small perturbation levels, these
reversals in $m_t$ values help improve the quality of further clusterings by
choosing a solution that is less efficient than that of standard HAC but more
consistent with the data (as already discussed in the example of
Figure~\ref{fig::outperforms}). Hence, when the data structure
is consistent with the constraint, $(m_t)_t$ typically provides an interpretable
dendrogram. This nice property is, of course, lost when the constraint is no
more consistent with the data structure (above a perturbation level of 60\%),
which is explained by the fact that the OCHAC has a poor performance in that
context, as already discussed in the previous section.

On the contrary, $(I_t)_t$ and $(\bar{I}_t)_t$ exhibit larger numbers of
reversals. This is particularly the case for the last mergers, even for small
levels of perturbation and even in the unconstrained case: 40-60\% of the
simulations have reversals for both OCHAC and standard HAC at a number of
clusters corresponding to the selected clustering. We also observe that the
percentage of simulations showing a reversal for standard HAC tends to decrease
when the perturbation level in the data increases for the first steps of the
hierarchical process (the same can be observed, to a much lesser extent, for
OCHAC). This phenomenon is explained below.
\begin{figure}[ht]
  \centering 
  \includegraphics[width=\linewidth]{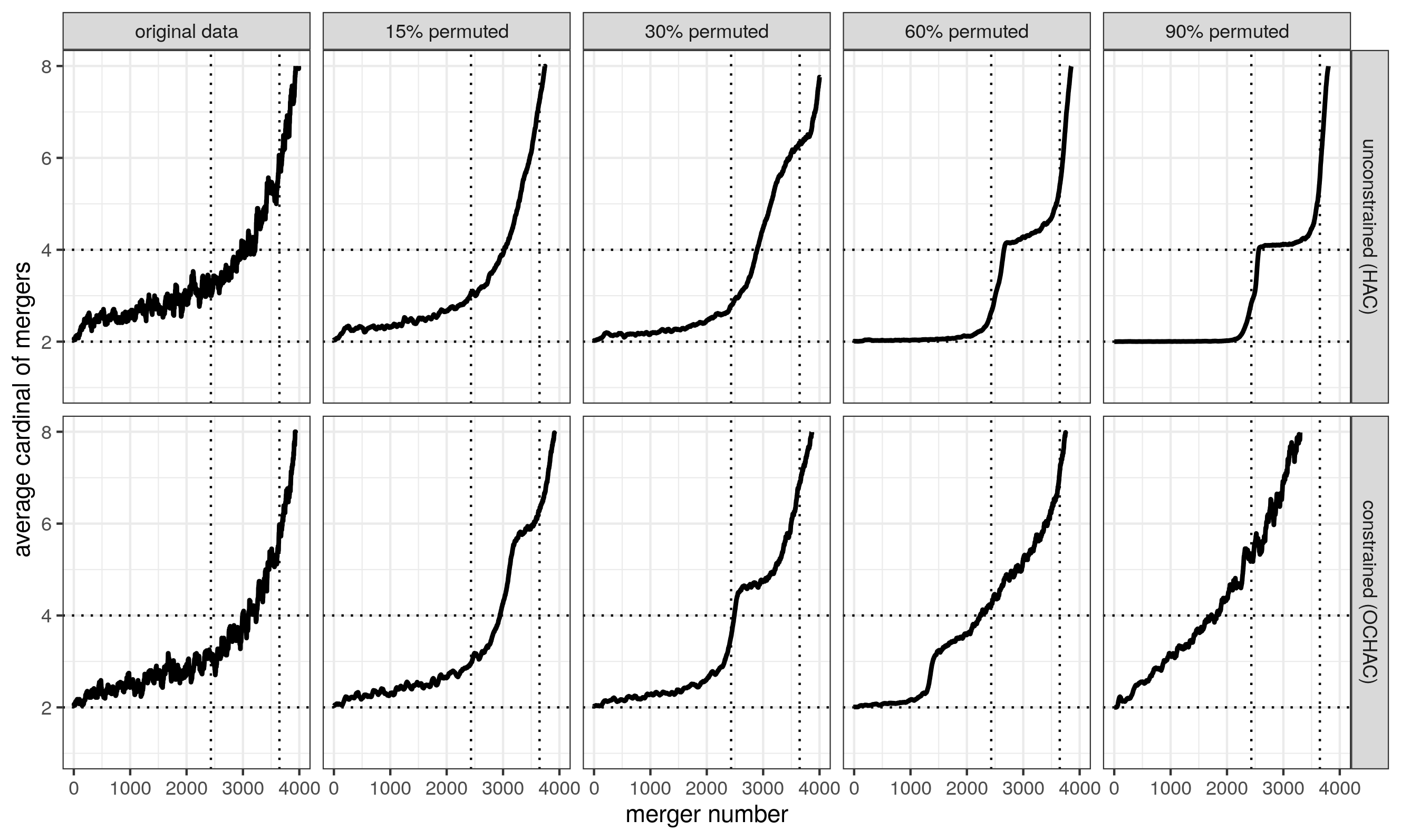}
  \caption{{\bf Evolution of the average cardinal of mergers along the 
hierarchical clustering process} for standard HAC (top) and OCHAC (bottom), and 
different levels of perturbation. Note that the average is computed over the 50 
simulations, whereas the original data correspond to a unique value. Data are 
shown only for the first 4,000 mergers and for a cardinal smaller than 8 for 
the sake of readability. The two dotted vertical lines correspond, respectively 
to $n/2$ and $3n/4$.}
  \label{fig::cardinal_distrib}
\end{figure}

Figure~\ref{fig::cardinal_distrib} displays the
evolution of the merged cluster size thorough the hierarchical clustering and
provides an explanation for this fact. For standard HAC, the number of clusters
with a size equal to 2 during the first steps of the algorithm is strongly
increasing when the perturbation level increases. For a permutation level of
90\%, most of the mergers have a size equal to 2 during half of the clustering
process (for fusion numbers ranging from 1 to at least 2,000). However, for
clusters with a size equal to 2, $I_t$ is equal to $m_t$ which explains the
similiarities between $m_t$ and $I_t$ curves during the first steps of the
clustering process, as the perturbation level increases. Since $(m_t)_t$ is
increasing for standard HAC, this explains why $I_t$ has less reversals in
standard HAC for the first merger numbers when the perturbation level is
higher. The same holds for $\bar{I}_t$ up to a fixed size factor of 2.

\section{Conclusion}

In this article, we studied the applicability of HAC and its constrained version to a wide range of input data. In particular, we have shown that these applications are  justified beyond the Euclidean framework. We have also shown that the monotonicity of the sequence of heights is not always ensured, although this property is necessary for the sequence of clusterings obtained by cutting dendrograms to be consistent with the sequence of clusterings of the algorithm. We have clarified which heights have this property depending on the input data types and for the constrained and unconstrained HAC. We have also pinpointed an important distinction between this monotonicity and the existence of crossovers.

These results imply that the variance of the merged cluster, $I_t$, or the
average variance of the merged cluster, $\bar{I}_t$, are never ensured to be
monotonic, and should thus not be chosen to represent the dendrogram
heights. Strikingly, we have also shown that the constrained version of the HAC
can provide more relevant and efficient solutions than its unconstrained
versions, not only in terms of algorithmic complexity, but also in terms of the
values of the objective function ESS$_t$. In such cases, a small number of
reversals can actually be beneficial to explore intermediate solutions closer to
the data and that lead to more relevant clusters.

\section*{Acknowledgements}

The authors would like to thank Marie Chavent for numerous instructive discussions on this paper.

The authors are grateful to the GenoToul bioinformatics platform (INRA Toulouse, http://bioinfo.genotoul.fr/) and its staff for providing computing facilities. 

\section*{Funding}
The PhD thesis of N.R. is funded by the INRA/Inria doctoral program 2018.
This work was also supported by the SCALES project funded by CNRS (Mission ``Osez l’interdisciplinarité'').

\bibliographystyle{apalike}
%\bibliography{BiblioCCHAC.bib}

\appendix

\section*{Appendix}
\addcontentsline{toc}{section}{Appendix}
\section{Proof of Proposition~\ref{prop::decreasing}}
\label{sec:proof-decreasing}
\begin{proof}[Proof of Proposition~\ref{prop::decreasing}]
 We begin by noting that by Proposition~\ref{lem:linkage-crossovers}, the only 
reversals that may occur are crossovers. With the notation of 
Proposition~\ref{prop::decreasing}, a crossover at step $t+1$ corresponds to the 
situation where
%  \begin{eqnarray*}
% \delta(G_l , G_r)\geq \min\left(\delta(G_l \cup G_r, G_{rr}), \delta(G_l \cup G_r, G_{ll})\right).
%  \end{eqnarray*}
 \begin{eqnarray*}
\delta(G_l , G_r)\geq \delta(G_l \cup G_r, G_{\bar{r}}) \textrm{ or } \delta(G_l , G_r)\geq  \delta(G_l \cup G_r, G_{\bar{l}}).
 \end{eqnarray*}
By symmetry we focus on the first case. With the notation of Proposition~\ref{prop::decreasing}, and using the Lance-Willams formula~\eqref{eq:lance-williams_ward}, the first condition is equivalent to
\begin{align*}
  \delta(G_l, G_r) \geq \frac{g_{lr'} \delta(G_l, G_{\bar{r}}) + g_{rr'} \delta (G_r, G_{\bar{r}})}{g_{lr'} + g_{rr'}}
\end{align*}
while the second one is equivalent to 
\begin{align*}
  \delta(G_l, G_r) \geq \frac{g_{\bar{l}l} \delta(G_{\bar{l}}, G_{l}) + g_{\bar{l}r} \delta (G_{\bar{l}}, G_{r})}{g_{\bar{l}l} + g_{\bar{l}r}}
\end{align*}
hence the result. \end{proof}

\section{Counter-example of the monotonicity of $\bar{I}_t$ for standard HAC 
in 
the Euclidean case}
\label{sec::counter_barI}

\begin{figure}[H]
  \centering 
  \begin{tabular}{m{0.5\linewidth} p{0.5\linewidth}}
    \includegraphics[width=\linewidth]{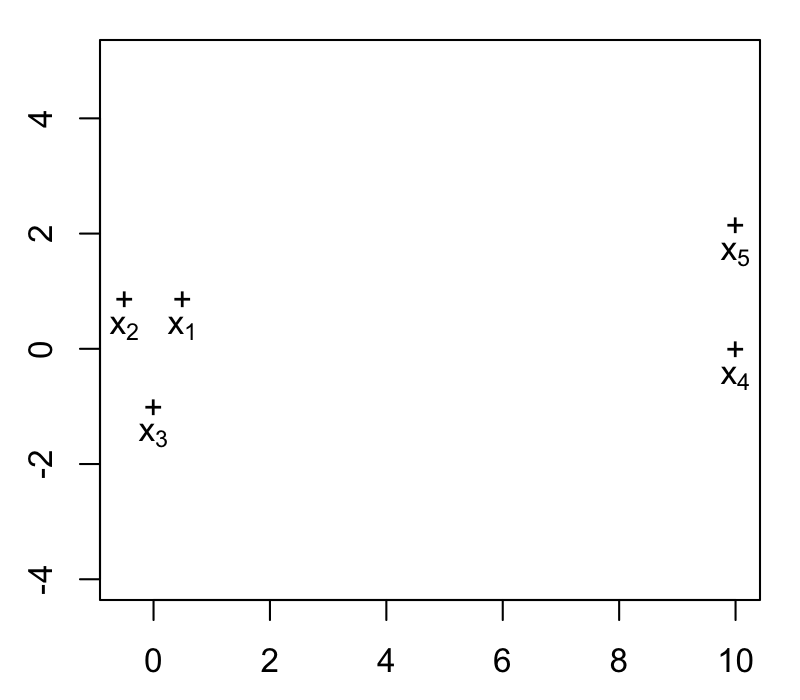} &
        \begin{minipage}{\linewidth}
        \footnotesize{
    $\left(\begin{array}{cc}
    x_1\\
    x_2\\
    x_3\\
    x_4\\
    x_5
  \end{array} \right) =\left(\begin{array}{cc}
    \frac{1}{2} & \frac{\sqrt{3}}{2} \\
    -\frac{1}{2} & \frac{\sqrt{3}}{2} \\
    0 & -1\\
    10 & 0\\
    10 & 2.15
  \end{array} \right)$\\
  $D \approx
\left(\begin{array}{ccccc}
 0.00 & 1.00 & 1.93 & 9.54 & 9.59 \\ 
1.00 & 0.00 & 1.93 & 10.54 & 10.58 \\ 
1.93 & 1.93 & 0.00 & 10.05 & 10.48 \\ 
9.54 & 10.54 & 10.05 & 0.00 & 2.15 \\ 
9.59 & 10.58 & 10.48 & 2.15 & 0.00 \\ 
\end{array}\right)$}
  \end{minipage}
  \end{tabular}\\
  \includegraphics[width=\linewidth]{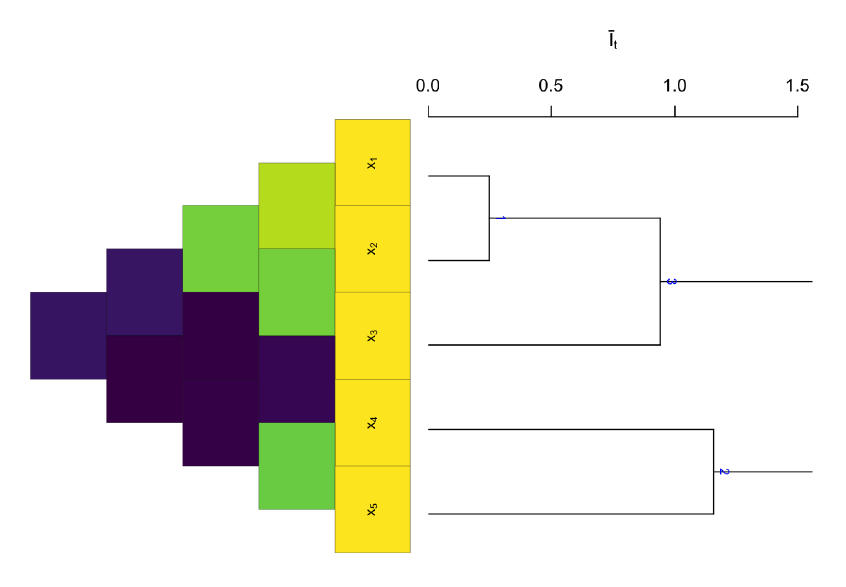}
  \caption{{\bf A reversal for Euclidean standard HAC with height defined as
      $\bar{I}_t$.} Top left: Configuration of the objects in $\mathbb{R}^2$. Top right: Coordinates of the objects and Euclidean distance matrix corresponding to
      this configuration. Bottom
    left: Representation of the values of the dissimilarity (dark colors
    correspond to larger values, so distant objects).  Bottom right: dendrogram
    obtained from standard HAC. Only the first 3 merges of the dendrogram is
    represented to ensure a comprehensive view of the sequence of heights.}
  \label{fig::AI}
\end{figure}

\section{Step-by-step description of the counter-examples}
\label{sec:counterexamples_standardHAC}

In the following tables, red color is used to signal reversals. Green color in details of Figure~\ref{fig::outperforms} is used to highlight the value of the objective function (ESS$_t$) for the clustering with 3 clusters.

\begin{table}[H]
\centering
\begin{tabular}{lllllll}
  \hline
  Merger & cluster 1 & cluster 2 & $m_t$ & ESS$_t$ & $I_t$ & 
$\overline{I}_t$ \\ 
  \hline
1 & $\{x_1\}$ & $\{x_2\}$ & 1.000 & 1.000 & 1.000 & 0.500 \\ 
  2 & $\{x_1,x_2\}$ & $\{x_3\}$ & \textcolor{red}{0.517} & 1.517 & 1.517 & 0.506 \\ 
   \hline
\end{tabular}
\caption{Details of Figure~\ref{fig::euclidean-reversal}}
\end{table}

\begin{table}[H]
\centering
\begin{tabular}{lllllll}
  &&OCHAC\\
  \hline
  Merger & cluster 1 & cluster 2 & $m_t$ & ESS$_t$ & $I_t$ & 
$\overline{I}_t$ \\ 
  \hline
1 & $\{x_1\}$ & $\{x_2\}$ & 2.500 & 2.500 & 2.500 & 1.250 \\ 
  2 & $\{x_1,x_2\}$ & $\{x_3\}$ & \textcolor{red}{2.167} & 4.667 & 4.667 & 1.556 \\ 
  3 & $\{x_6\}$ & $\{x_7\}$ & 2.500 & 7.167 & 2.500 & 1.250 \\ 
  4 & $\{x_5\}$ & $\{x_6,x_7\}$ & \textcolor{red}{2.167} & \textcolor{green}{9.333} & 4.667 & 1.556 \\ 
  5 & $\{x_1, x_2, x_3\}$ & $\{x_4\}$ & 13.333 & 22.667 & 18.000 & 4.500 \\ 
  6 & $\{x_1, x_2, x_3, x_4\}$ & $\{x_5, x_6, x_7\}$ & 20.762 & 43.429 & 43.429 & 6.204 \\ 
     \hline
          \\
     &&HAC\\
      \hline
  Merger & cluster 1 & cluster 2 & $m_t$ & ESS$_t$ & $I_t$ & 
$\overline{I}_t$ \\ 
     \hline
1 & $\{x_2\}$ & $\{x_6\}$ & 0.500 & 0.500 & 0.500 & 0.250 \\ 
  2 & $\{x_1\}$ & $\{x_3\}$ & 2.000 & 2.500 & 2.000 & 1.000 \\ 
  3 & $\{x_5\}$ & $\{x_7\}$ & 2.000 & 4.500 & 2.000 & 1.000 \\ 
  4 & $\{x_2,x_6\}$ & $\{x_1,x_3\}$ & 6.250 & \textcolor{green}{10.750} & 8.750 & 2.188 \\ 
  5 & $\{x_4\}$ & $\{x_1,x_2,x_3,x_6\}$ & 13.250 & 24.000 & 22.000 & 4.400 \\ 
  6 & $\{x_5,x_7\}$ & $\{x_1,x_2,x_3,x_4,x_6\}$ & 19.429 & 43.429 & 43.429 & 6.204 \\ 
   \hline
\end{tabular}
\caption{Details of Figure~\ref{fig::outperforms}}
\end{table}

\begin{table}[H]
\centering
\begin{tabular}{lllllll}
  \hline
Merger & cluster 1 & cluster 2  & $m_t$ & ESS$_t$ & $I_t$ & 
$\overline{I}_t$ \\ 
  \hline
1 & $\{x_1\}$ & $\{x_2\}$ & 0.50 & 0.50 & 0.50 & 0.25 \\ 
  2 & $\{x_1,x_2\}$ & $\{x_3\}$ & 2.32 & 2.82 & 2.82 & 0.94\\ 
  3 & $\{x_4\}$ & $\{x_5\}$ & 2.33 & 5.15 & \textcolor{red}{2.33} & 1.17 \\ 
  4 & $\{x_1,x_2,x_3\}$ & $\{x_4,x_5\}$ & 120.84 & 125.99 & 125.99 & 25.20 \\ 
   \hline
\end{tabular}
\caption{Details of Figure~\ref{fig::I}}
\end{table}

\begin{table}[H]
\centering
\begin{tabular}{lllllll}
  \hline
 Merger & cluster 1 & cluster 2  & $m_t$ & ESS$_t$ & $I_t$ & 
$\overline{I}_t$ \\ 
  \hline
1 & $\{x_1\}$ & $\{x_2\}$ & 0.995 & 0.995 & 0.995 & 0.498 \\ 
  2 & $\{x_1,x_2\}$ & $\{x_3\}$ & 0.998 & 1.993 & 1.993 & 0.664 \\ 
  3 & $\{x_1,x_2,x_3\}$ & $\{x_4\}$ & \textcolor{red}{ 0.997} & 2.990 & 2.990 & 0.748 \\ 
  4 & $\{x_1,x_2,x_3,x_4\}$ & $\{x_5\}$ & \textcolor{red}{-0.192} & \textcolor{red}{2.798} & \textcolor{red}{2.798} & \textcolor{red}{0.560} \\ 
  5 & $\{x_1,x_2,x_3,x_4,x_5\}$ & $\{x_6\}$ & 0.534 & 3.332 & 3.332 & \textcolor{red}{0.555} \\ 
   \hline
\end{tabular}
\caption{Details of Figure~\ref{fig::ESS}}
\end{table}

\begin{table}[H]
\centering
\begin{tabular}{lllllll}
  \hline
 Merger & cluster 1 & cluster 2  & $m_t$ & ESS$_t$ & $I_t$ & 
$\overline{I}_t$ \\ 
  \hline
1 & $\{x_1\}$ & $\{x_2\}$ & 0.50 & 0.50 & 0.50 & 0.25 \\ 
  2 & $\{x_4\}$ & $\{x_5\}$ &  2.31 & 2.81 & 2.31 & 1.16 \\ 
  3 & $\{x_1,x_2\}$ & $\{x_3\}$ &  2.32 & 5.13 & 2.82 & \textcolor{red}{0.94} \\ 
  4 & $\{x_1,x_2,x_3\}$ & $\{x_4,x_5\}$ &  120.83 & 125.96 & 125.96 & 25.19 \\ 
   \hline
\end{tabular}
\caption{Details of Figure~\ref{fig::AI}}
\end{table}

\end{document}